\documentclass[article]{imsart} 

\RequirePackage{amsthm,amsmath}
\usepackage{amsfonts}
\usepackage{amsmath}
\usepackage{dsfont}
\usepackage{xcolor} 
\usepackage{graphicx}
\usepackage{subcaption}
\usepackage{algorithm2e}
\usepackage{float}
\usepackage{placeins}

\RestyleAlgo{ruled}

\RequirePackage[colorlinks,citecolor=blue,urlcolor=blue]{hyperref}

\arxiv{https://arxiv.org/abs/2307.12754}

\startlocaldefs
\numberwithin{equation}{section}
\theoremstyle{plain}
\newtheorem{theorem}{Theorem}[section]
\newtheorem{lemma}{Lemma}[section]
\newtheorem{definition}{Definition}[section]
\newtheorem{corollary}{Corollary}[section]
\newtheorem{assumption}{Assumption}[section]
\DeclareMathOperator*{\argmax}{arg\,max}
\DeclareMathOperator*{\argmin}{arg\,min}
\DeclareMathOperator*{\trace}{tr}
\DeclareMathOperator{\Diag}{Diag}
\DeclareMathOperator{\Vect}{Span}
\definecolor{r}{rgb}{0.89, 0.02, 0.17}
\newcommand{\note}[1]{{#1}} 
\endlocaldefs

\begin{document}

\begin{frontmatter}

\title{Nonparametric Linear Feature Learning in Regression Through Regularisation}
\runtitle{Nonparametric Linear Feature Learning}

\begin{aug}
\author{\fnms{Bertille} \snm{Follain}\ead[label=e1]{bertille.follain@inria.fr}}

\address{Inria, PSL Research University\\ 48 rue Barrault, 75013 Paris, France\\
\printead{e1}}

\author{\fnms{Francis} \snm{Bach}\ead[label=e2]{francis.bach@inria.fr}}

\address{Inria, PSL Research University\\ 48 rue Barrault, 75013 Paris, France\\
\printead{e2}}
\runauthor{B. Follain \& F. Bach}

\end{aug}

\begin{abstract}
Representation learning plays a crucial role in automated feature selection, particularly in the context of high-dimensional data, where non-parametric methods often struggle. In this study, we focus on supervised learning scenarios where the pertinent information resides within a lower-dimensional linear subspace of the data, namely the multi-index model. If this subspace were known, it would greatly enhance prediction, computation, and interpretation. To address this challenge, we propose a novel method for \note{joint linear feature learning and non-parametric function estimation, aimed at more effectively leveraging hidden features for learning}. Our approach employs empirical risk minimisation, augmented with a penalty on function derivatives, ensuring versatility. Leveraging the orthogonality and rotation invariance properties of Hermite polynomials, we introduce our estimator, named \textbf{RegFeaL}. By using alternative minimisation, we iteratively rotate the data to improve alignment with leading directions. \note{We establish that the expected risk of our method converges in high-probability to the minimal risk under minimal assumptions and with explicit rates}. Additionally, we provide empirical results demonstrating the performance of \textbf{RegFeaL} in various experiments.
\end{abstract}


\tableofcontents
\end{frontmatter}

\section{Introduction}
The increasing availability of high-dimensional data has created a demand for effective feature selection methods that can handle complex datasets. Representation learning, which aims to automate the feature selection process, plays a crucial role in extracting meaningful information from such data. However, non-parametric methods often struggle in high-dimensional settings. 

A sensible approach is to consider that there are a lower number of unknown relevant linear features, or linear transformations of the original data, that explain the relationship between the response and factors. A popular way to model this is to consider the multi-index model \cite{multi_index_model}, where we assume that the prediction function is the composition of few linear features which form a linear subspace (the effective dimension reduction (e.d.r.) subspace) and a non-parametric function. The multi-index model has been used in practice in many fields, such as ecology \cite{plague} or bio-informatics \cite{10.1093/bioinformatics/btg062}.  If the features were known, learning would be much easier due to the lower dimensionality of the problem, and their low number allows for a simpler, more explainable model, as well as a lesser need for computational and storage resources. 
\note{Although these relevant features are not known a priori, recognising their existence enables the development of methods that incorporate them, potentially resulting in better estimators for prediction.}

\paragraph{Related work.}
A wide range of methods have been proposed to estimate the e.d.r.\ space in the context of multi-index models. Brillinger introduced the method of moments, initially designed for Gaussian data and an e.d.r.\ of dimension one \cite{brillinger}. This method uses specific moments to eliminate the unknown function and focuses solely on the influence of the e.d.r.\ space. Extensions of this approach for distributions with differentiable log-densities have been provided, resulting in the average derivative estimation (ADE) method \cite{ADE}.

To incorporate subspaces of any dimension, several methods have been proposed. Slicing methods, such as slice inverse regression (SIR) \cite{sir}, use second-order moments to account for subspaces. Principal Hessian directions (PHD) \cite{phd} extend the approach to elliptically symmetric data. Combining these techniques, sliced average derivative estimation (SADE) \cite{babichev:hal-01388498} offers a comprehensive approach. However, these methods heavily rely on assumptions about the distribution shape and require prior knowledge of the distribution, limiting their applicability.

Iterative improvements have been suggested for both the one-dimensional latent subspace case \cite{single_index} and the general case \cite{dalalyan:hal-00128129}. Other optimisation-based methods, such as local averaging, aim to minimise an objective function to estimate the subspace \cite{kernel_dimension_reduction_francis, MAVE}. Although these procedures exhibit favourable performance in practice, particularly the \textbf{MAVE} method \cite{MAVE}, the theoretical guarantees provided by \cite{MAVE} show exponential dependency in the dimension of the original data. \note{Nonetheless, the recent work by \cite{review}  has made significant contributions to sufficient dimension reduction (SDR) by providing robust theoretical results for high-dimensional data that do not exhibit exponential dependency. However, their method, designed primarily for dimension reduction and variable selection in the specific setting of the square loss, relies on the linearity condition, which holds for example under the assumption that the covariates follow an elliptically contoured distribution.}

In our work, we consider regularising the empirical risk by incorporating derivatives, a technique employed in various contexts. Classical splines, such as Sobolev spaces regularisation \cite{wahba}, have used derivative-based regularisation. More recently, derivative regularisation has been employed in the context of semi-supervised learning \cite{cabannes2021overcoming}, as well as in linear subspace estimation using SADE~\cite{babichev:hal-01388498}.

\paragraph{Contributions.}
\note{We propose a novel approach for joint function estimation and effective dimension reduction space estimation in multi-index models.}

\note{We employ the empirical risk minimisation framework, compatible with a wide range of loss functions, which is regularised by a penalty on the derivatives of the prediction function. The proposed regularisation enforces dependence on a reduced set of projected dimensions. Our method addresses the discussed limitations of previous methods. Indeed the assumptions on the distribution of the covariates are minimal (typically subgaussianity of the norm), and does not require said distribution to be known a priori. We are also able to provide explicit rates for the high-probability convergence of the expected risk of our estimator to the minimal risk, again with limited assumptions.}

\note{To construct our estimator, which we coin \textbf{RegFeaL}, we exploit the advantageous properties of Hermite polynomials, which exhibit orthogonality and rotation invariance. By incorporating alternative minimisation on a variational formulation of the problem, we enable iterative rotation of the data to better align with the leading directions, as well as easy computation of the unknown relevant dimension of the e d.r. space. Furthermore, for the specific case of the variable selection problem, that is, when only a subset of the coordinates of the original data is relevant, we can simplify our proposed penalty term which yields a computationally more efficient algorithm.}

\note{While our primary objective is to leverage the existence of a dependency on only a few variables or features, we also offer principled ways to estimate the dimension of the feature space and select the relevant features.}

\note{We provide detailed explanations about the efficient computation of our estimator, ensuring its practical usability. Additionally, we present theoretical results that establish the high-probability convergence to the minimal risk of the expected risk of our estimator, with limited assumptions on the loss and data distribution. This allows for a deeper understanding of the performance of the method and the dependency on certain parameters such as the dimension of the original data and the number of samples.}

\note{To demonstrate the strengths of our approach, we conduct an extensive set of experiments focusing on training behaviour, dependency on sample size and dimension, and comparison to other methods.}

\note{Importantly, our regularisation strategy is applicable to a wide range of problems where empirical risk can be formulated, making it a versatile tool for feature learning and dimensionality reduction tasks, potentially extending beyond statistics to fields such as signal processing and control.}

\note{In summary, our contributions encompass the introduction of a novel empirical risk minimisation framework with derivative-based regularisation for prediction and e.d.r. space estimation in multi-index models. We provide efficient computational techniques, theoretical insights, and empirical evidence, highlighting the advantages of our proposed method.}

\paragraph{Paper organisation.}
The paper is organised as follows: we begin by describing the problem, our penalties, and the use of Hermite polynomials in Section~\ref{sec:preliminaries}. Then, we address the question of effectively computing our estimator \textbf{RegFeaL} in Section~\ref{sec:opti}. In Section~\ref{sec:statistical_consistency}, we discuss the \note{convergence of the empirical risk} of our estimator. In Section~\ref{sec:num}, we present numerical studies to illustrate the behaviour of \textbf{RegFeaL}. Finally, in Section~\ref{sec:conclusion}, we summarise our findings, highlight the contributions of our research, and discuss potential future directions.

\paragraph{Notations.}
Let $\mathbb{N}$ denote the set of non-negative integers and $\mathbb{N}^*$ the set of positive integers. For $d \in \mathbb{N}$, let $[d] = { 1, \ldots, d}$. Given $x\in \mathbb{R}^d$ and $a\in [d]$, $x_a$ represents the $a$-th component of $x$. Similarly, for $S\subset [d]$, $x_S$ denotes $(x_a)_{a \in S}$. Let $p,d \in \mathbb{N}^*$, and consider a matrix $A \in \mathbb{R}^{p \times d}$. The matrix $A_S$ corresponds to the columns of $A$ extracted using indices from $S$, while $A_{i,j}$ represents the element of $A$ in the $j$-th position of row $i$. The cardinality of a set $S$ is denoted by $|S|$. $I_d$ represents the $d\times d$ identity matrix, and $O_d$ denotes the set of $d\times d$ orthogonal matrices. For any $d \times d$ matrix $A$, $\trace(A)$ denotes its trace, and $\Diag(A)$ represents the diagonal matrix of size $d\times d$ with the diagonal elements of $A$. The transpose of a matrix $B$ is denoted by $B^\top$. For an invertible matrix $\Lambda$, $\Lambda^{-1}$ represents its inverse. Given $\eta \in \mathbb{R}^d$, $\Diag(\eta)$ is the diagonal matrix of size $d\times d$ with $\eta$ as its diagonal. For $r>0$, $\|\eta\|_r = \big(\sum_{a=1}^d |\eta_a|^r\big)^{1/r}$. For any $\alpha \in \mathbb{N}^d, \ |\alpha| = \sum_{a=1}^d \alpha_a$.

\section{Preliminaries}
\label{sec:preliminaries}
\subsection{Problem description}
We consider a standard regression problem, where we have access to a dataset $(x^{(i)}, y^{(i)})_{i \in [n]}, n \in \mathbb{N}^*$ consisting of independent and identically distributed (i.i.d.) realisations of a pair of random variables~$(X, Y)$ with probability measure~$\note{\nu}$ on~$\mathcal{X} \times \mathcal{Y} \subset \mathbb{R}^d \times \mathbb{R}$. Our objective is to estimate the regression function~$f^* := \argmin_{f \in \mathcal{F}} \mathcal{R}(f)$, where~$\mathcal{R}(f) := \mathbb{E}_{\note{\nu}}(\ell(Y, f(X)))$ is the risk,~$\ell$ is a loss function and~$\mathcal{F}$ a space of functions from~$\mathbb{R}^d$ to~$\mathbb{R}$. At this stage, we do not impose any assumptions regarding the choice of loss function or the data distribution.

We consider the multi-index model \cite{multi_index_model}, i.e., a model where the regression function depends on a low-rank linear transformation of the original variables.
\begin{assumption}[Feature learning] \label{ass:feature}
We assume that the regression function~$f^*$ can be expressed as the combination of a rank $s$ linear transformation~$P$ and a function~$g^*$ from~$\mathbb{R}^{s} \to \mathbb{R}$, i.e., 
\begin{equation*}
\exists s \in [d], \ \exists P \in \mathbb{R}^{d \times s}, \ P^\top P =I_{s},\  \exists g^* : \mathbb{R}^s \to \mathbb{R},  \ \forall x \in \mathbb{R}^d, \ f^*(x) = g^*(P^\top x).
\end{equation*}
\end{assumption}
We do not assume any prior knowledge about the value of $s$.  The model is nonparametric hence it remains broad. Our objective is to simultaneously estimate both $f^*$ and the associated linear transformation $P$, as well as the dimension $s$, by means of regularised empirical risk minimisation. Recall the definition of the empirical risk~$ \widehat{\mathcal{R}}(f):= \frac{1}{n} \sum_{i=1}^n \ell(y^{(i)}, f(x^{(i)}))$. This approach offers versatility, allowing its application to various scenarios. Although our focus lies on the regression setting, we acknowledge the potential of the regularisation-based method for future work in any setting where a risk can be defined.

\subsection{Penalising by derivatives}
\label{sec:derivatives}
In this context, it is common to employ derivative-based regularisation techniques \cite{babichev:hal-01388498, og_rosasco}. Under mild regularity assumptions, if we express~$f$ as~$f = g(Q^\top \cdot)$ with~$Q \in \mathbb{R}^{d \times s}$, then for all~$ x \in \mathbb{R}^d, \ \nabla f(x) \cdot \nabla f(x)^\top  = Q \nabla g(x) \cdot \nabla g(x)^\top  Q^\top$, where~$\nabla f(x) \in \mathbb{R}^d$ denotes the gradient of~$f$ at point~$x$. Consequently, we observe that
\begin{equation*}
    \int_{\mathcal{X}} \nabla f \nabla f^\top \note{\nu}  = \bigg(\int_{\mathcal{X}} \frac{\partial f}{\partial x_a} \frac{\partial f}{\partial x_b} \note{\nu} \bigg)_{a, b \in [d]}
\end{equation*}
has a rank of at most~$s$. This observation motivates us to employ the rank  of~$ \int_{\mathcal{X}} \nabla f \nabla f^\top \note{\nu}$ as a penalisation. However, the discontinuous nature of the rank makes this approach challenging for optimisation. To address this, we could penalise  instead by~$ \trace\big(\int_{\mathcal{X}} \nabla f \nabla f^\top \note{\nu}\big)$ as a convex relaxation~\cite{recht_trace_norm}.

This strategy would extend the work of \cite{og_rosasco}, which focuses on variable selection, a special case of feature learning. It corresponds to the constraint that~$P$ from Assumption~\ref{ass:feature} only contains~$0$ and~$1$ (with exactly a single one in each column), resulting in a model where the regression function depends on a limited number of the original variables.
\begin{assumption}[Variable selection] \label{ass:variable}
We assume that~$f^*$, the regression function, actually only depends on~$s$ of the~$d$ variables, i.e., 
\begin{equation*}
\exists s \in [d], \ \exists S \subset [d], \ |S| = s, \ \exists \ g^* : \mathbb{R}^{s} \to \mathbb{R}, \ \forall x \in \mathbb{R}^d, \ f^*(x) = g^*(x_{S}).
\end{equation*}
\end{assumption}
\note{In this variable selection setting}, we can remark that it suffices to penalise by a simpler quantity. Specifically, under some mild regularity assumptions on the function~$f$, $f$ does not depend on variable~$x_a$ if and only if the partial derivative of~$f$ with respect to~$x_a$, denoted by~$\frac{\partial f}{\partial x_a}$, is null everywhere on~$\mathcal{X}$. Hence, the task is to design a penalty that enforces sparsity in the dependence on different variables.

\note{To address this, we can draw inspiration from the group Lasso \cite{group_lasso}, which extends the Lasso method to enable structured sparsity. The group Lasso encourages groups of related quantities to be selected or excluded together by penalising the sum over each group using an appropriate penalty. For example, the derivatives with regard to a variable $x_a$ at data points $x^{(i)}$ should all be null if the function does not depend on variable $x_a$. Hence, they constitute a relevant group for group Lasso.} 

Combining these observations, \cite{og_rosasco} proposed a strategy using the fact that for all $a \in [d]$, $f$ does not depend on $x_a$ if and only if $\int_{\mathcal{X}} \left(\frac{\partial f}{\partial x_a}(x)\right)^2 \note{\nu} = 0$. They introduced penalties on each variable and summed them to obtain the penalty $  \sum_{a=1}^d \big(  \int_{\mathcal{X}} \big(\frac{\partial f}{\partial x_a}(x) \big)^2 \note{\nu}(x) \mathrm{d}x \big)^{1/2}$. However, since these quantities are intractable due to the unknown nature of $\note{\nu}$, they use a data-dependent penalty instead
\begin{equation*}
    \sum_{a=1}^d \left(\frac{1}{n} \sum_{i=1}^n \left(\frac{\partial f}{\partial x_a}(x^{(i)})\right)^2\right)^{1/2}.
\end{equation*}
By assuming that~$f$ belongs to some regular reproducing kernel Hilbert space (RKHS), the partial derivatives are easily computable, and so is the penalty \cite{og_rosasco} (for a good introduction to RKHS, see \cite{aronszajn50reproducing}). However, this regularisation by an estimate of the~$L^2$ norms of derivatives in the context of RKHS is not suitable. Functions that depend on a single variable, such as $x_1$, do not belong to the RKHS, making it an inappropriate space for addressing this type of problem. Additionally, another regularisation by the norm in the RKHS is required, introducing an extra hyperparameter. Moreover, using derivatives only at the data points limits the exploitation of the power of regularity.

We are confronted with two challenges here. First, how can the penalisation scheme be improved for variable selection? Second, how can it be adapted for feature learning? While our primary goal is the latter, we consider the former as a by-product of our methodology.

To address both challenges, we employ Hermite polynomials \cite{hermite_2009}, although it is worth noting that various other alternatives could have been considered for the first problem where rotation invariance is not needed.

\subsection{Hermite polynomials for variable selection}
To facilitate understanding, let us first consider the simpler case of variable selection. We employ multidimensional Hermite polynomials due to their suitability for both variable selection and feature learning. The normalised one-dimensional Hermite polynomials~$(h_k(x))_{k \geq 0}$ form an  orthonormal polynomial basis for the standard Gaussian measure on~$\mathbb{R}$ with density~$\frac{1}{\sqrt{2\pi}} e^{-x^2/2}$. The first few polynomials are given by\footnote{Given the regular ``physicist'' Hermite polynomials~$H_k$ (not to be confused with multivariate polynomials defined in Equation~\eqref{eq:multivariate-hermite}), we have~$h_k(x) = \frac{1}{\sqrt{2^k k!}} H_k(x/\sqrt{2})$ for any~$k \in \mathbb{N}$ and  for the ``probabilist'' Hermite polynomials $He_k$, we have~$h_k(x) = \frac{1}{\sqrt{n!}}He_k(x)$.}
\begin{align*}
h_0(x)  =  1, \ h_1(x)  =  x, \ h_2(x)  =  \frac{1}{\sqrt{2}} ( x^2 - 1 ), \ h_3(x)  =  \frac{1}{\sqrt{6}} ( x^3 - 3x ).
\end{align*}
These polynomials possess useful properties that allow their recursive computation and characterise their growth and their derivatives
\begin{align}
h_{n+2}(x) & = \frac{x}{\sqrt{n+2}}\cdot h_{n+1}(x) - \sqrt{\frac{n+1}{n+2}} \cdot h_n(x) \\
h_n'(x) & =  \sqrt{n} \cdot h_{n-1}(x)\label{eq:hermite_reccurence_relation_derivative} \\
|h_n(x)| & \leq  \exp(x^2/4). \label{eq:hermite_bound} 
\end{align}
The last property can be proved using  Hermite functions and Cramer's inequality \cite{szegő1939orthogonal}.

Next, we define the multivariate polynomials as follows
\begin{equation}
\label{eq:multivariate-hermite}
\big(H_\alpha \big)_{\alpha \in \mathbb{N}^d} \text{ where } \forall x \in \mathbb{R}^d, \ H_\alpha(x) = \prod_{a=1}^d h_{\alpha_a} (x_a).
\end{equation} 
This family forms an orthonormal basis of the space~$L^2(q) := \big\{f : \mathbb{R}^d  \to \mathbb{R}, \ \int_{\mathbb{R}^d} f^2q < + \infty  \big\}$ where~$ q(x) = \frac{1}{(2\pi)^{d/2}} e^{-\|x\|^2/2}$ denotes the standard normal distribution on~$\mathbb{R}^d$. We now present a Lemma which justifies the use of the multivariate Hermite polynomials in the variable selection setting.

\begin{lemma}[\note{Equivalence for dependency on variables}]
\label{lem:sparse_variable}
Let~$f \in L^2(q)$ and express it as~$f = \sum_{\alpha \in \mathbb{N}^d} \hat{f}(\alpha) H_\alpha$. Then for any~$b \in [d]$,
\begin{equation*}
f \text{ does not depend on variable } x_b \iff   \forall \alpha \in \big(\mathbb{N}^d\big)^*, \ \alpha_b \neq 0 \implies \hat{f}(\alpha) =0.
\end{equation*}
\end{lemma}

\begin{proof}[\note{Proof of Lemma~\ref{lem:sparse_variable}}]
For~$x \in \mathbb{R}^d$, we have $h_0(x)=1$ and
\begin{align*}
f(x) = &\underbrace{\hat{f}(0) + \sum_{\alpha \in (\mathbb{N}^d)^*,\ \alpha_b =0} \hat{f}(\alpha) \prod_{a\in [d] \setminus \{b\}} h_{\alpha_a}(x_a)} _{\text{does not depend on} \ x_b} \\ 
& \hspace*{4cm} + \underbrace{\sum_{\alpha \in (\mathbb{N}^d)^*, \ \alpha_b >0} \hat{f}(\alpha) \prod_{a\in [d]} h_{\alpha_a}(x_a)}_{\text{depends on} \ x_b},
\end{align*} 
i.e., $f$ can be decomposed into two additive components, one of which does not depend on~$x_b$. For the component that depends on~$x_b$, it is the sum over~$\alpha \in \mathbb{N}^d$ such that~$\alpha_b$ is nonzero, yielding the result.
\end{proof}

We observe that when $f$ does not depend on a variable, it corresponds to a specific sparsity pattern in the coefficients $\hat{f}(\alpha)$ with respect to the basis $(H_\alpha)_{\alpha \in \mathbb{N}^d}$. \note{Indeed, if $f$ does not depend on $x_b$, all coefficients $\hat{f}(\alpha) $ for $ \alpha $ in the group $\{ \alpha \in \big(\mathbb{N}^d\big)^*, \alpha_b > 0 \}$ must be null. These groups overlap for different variables, and a similar argument holds for feature learning as we will see in Section~\ref{sec:feature}.} \note{This specific sparsity pattern motivates the use of a penalty based on group Lasso \cite{group_lasso}, and more specifically overlapping group Lasso~\cite{jenatton2011structured}.}

Hence, the Hermite polynomial basis is well-suited to this variable selection setting, while the space $L^2(q)$ is sufficiently large to describe a wide range of functions. However, it is worth noting that other spaces and well-adapted bases, such as any orthonormal basis of square-integrable functions, could also be used. Moreover, we use the Gaussian measure only to define the basis, and our method can be applied to all distributions.

To define a penalty relevant to variable selection, we examine the derivatives of $H_\alpha$. Here, we decompose any~$f \in \mathcal{F}$  as~$f = \sum_{\alpha \in \mathbb{N}^d} \hat{f}(\alpha) H_\alpha$.  Let~$e_a$ denote the~$a$-th element of the canonical basis of~$\mathbb{R}^d$, for~$a \in [d]$. Using Equation~\eqref{eq:hermite_reccurence_relation_derivative}, we obtain the following identities
\begin{align}
\frac{\partial H_\alpha}{\partial x_a}  &= \sqrt{\alpha_a} H_{\alpha - e_a} \\
\frac{\partial f}{\partial x_a}  &= \sum_{\alpha \in (\mathbb{N}^d)^*}\sqrt{\alpha_a} \hat{f}(\alpha) H_{\alpha - e_a} \label{eq:partial_f} \\
\int_{\mathbb{R}^d} \bigg(\frac{\partial f}{\partial x_a} \bigg)^2 q &= \sum_{\alpha \in (\mathbb{N}^d)^*}\alpha_a \hat{f}(\alpha)^2 .\label{eq:derivatives_a}
\end{align}

However, we remark that Equation~\eqref{eq:derivatives_a} corresponds to the expected version of the penalty proposed by \cite{og_rosasco} (when~$\note{\nu} = q$), which we deemed not suitable for our problem: indeed, penalising the \note{$L^2$}-norm of derivatives does not impose enough regularity for statistically efficient non-parametric estimation and thus requires extra regularisation, as specified by \cite{og_rosasco}.

We consider instead introducing a sequence~$(c_k)_{k > 0}$ of non-negative reals, to further regularise and avoid the need for additional regularisation.  We consider the space~$\mathcal{F}$, spanned by the family composed of~$H_\alpha$ for~$\alpha =0$ or~$\alpha \in (\mathbb{N}^d)$ such that~$c_{|\alpha|}>0$, i.e.,~$  \mathcal{F} := \Vect \big( \{H_0\} \cup \{H_\alpha,  \text{ for } \alpha \in (\mathbb{N}^d)^* \text{ such that } c_{|\alpha|>0} \} \big)$  and consider two penalties. First, we define a sparsity-inducing penalty, which depends on a hyper-parameter~$r \in (0, +\infty)$
\begin{equation*}
\Omega_{\mathrm{var}} (f)  =  \bigg(\sum_{a = 1}^d \bigg( \sum_{\alpha \in (\mathbb{N}^d)^*} \!\! \alpha_a \frac{1}{c_{|\alpha|}} \hat{f}(\alpha)^2 \bigg)^{r/2}\bigg)^{1/r}.
\end{equation*}
This penalty encourages sparsity in the dependence of $f$ on individual variables, as it pushes quantities of the form~$ \big( \sum_{\alpha \in (\mathbb{N}^d)^*} \alpha_a \frac{1}{c_{|\alpha|}} \hat{f}(\alpha)^2 \big)^{r/2}~$ to be~$0$. When this is the case, we obtain that~$ \forall \alpha \in (\mathbb{N}^d)^*, \ \alpha_a \neq 0, \hat{f}(\alpha) =0$, i.e.,~$f$ does not depend on variable~$x_a$  (Lemma~\ref{lem:sparse_variable}). When~$r\geq 1$,~$\Omega_{\rm var}$ is a norm, which makes the problem easier to study from a theoretical point of view because if the loss is convex, this will yield a convex optimisation problem. \note{However, estimators obtained through regularised empirical risk minimisation often suffer from bias due to the strong shrinkage associated with sparsity. Convex penalties can inadvertently reduce the significance of essential variables or features by excessive shrinkage to enforce sparsity. To address these issues, one can retrain on the set of selected variables or use concave penalties, which, despite presenting more analytical challenges, frequently deliver superior results by pushing the solution towards the boundary and enhancing sparsity \cite{zhang2010nearly, livreviolet}. In this work, we adopt this strategy through the hyper-parameter $r$ when $r<1$, which is the choice used in practice, while $r=1$ is used in the theoretical analysis.}

The link with the nullity of the derivative can be seen using Equation~\eqref{eq:derivatives_a}
\begin{equation*}
\bigg( \sum_{\alpha \in (\mathbb{N}^d)^*}\!\! \alpha_a \frac{1}{c_{|\alpha|}} \hat{f}(\alpha)^2 \bigg)^{r/2} = 0\  \iff \  \int_{\mathbb{R}^d} \bigg(\frac{\partial f}{\partial x_a} \bigg)^2 q =0.    
\end{equation*}

Next, we introduce a smoothness-inducing norm, which penalises higher-order polynomials, i.e., those with large~$|\alpha|$ (the dependence only on $|\alpha|$ is needed for future rotation invariance)
\begin{equation*}
        \Omega_0(f) =  \bigg( \sum_{\alpha \in (\mathbb{N}^d)^*}  \! \frac{1}{c_{|\alpha|}} \hat{f}(\alpha)^2 \bigg)^{1/2}.
\end{equation*}
It is important to note that $\Omega_0$ is not integrated into the theoretical analysis and will be used with a much smaller and fixed parameter compared to~$\Omega_{\rm var}$. Its primary purpose is to enforce numerical stability during the optimisation procedure, as discussed in Section \ref{sec:opti}.

The choice of~$(c_k)_{k \in \mathbb{N}^*}$ significantly influences the behaviour of the penalties.  In this work, we will consider two specific choices:~$c_k = \mathds{1}_{k \leq M}$ for some~$M \in \mathbb{N}$ and~$c_k = \rho^k$ for some~$\rho \in [0,1)$. Both choices ensure that all three penalties are well-defined. Notably, when~$M=1$,~$\Omega_{\mathrm{var}}$ considered with the quadratic loss reduces to the \note{basic Lasso problem with linear features} \cite{lasso}.

It is worth mentioning that the coefficient $\hat{f}(0)$, which corresponds to the constant function $H_0 = 1$, is never penalised because it does not depend on any of the variables.

We then consider estimating~$f^*$ in the setting described in Assumption~\ref{ass:variable} by
\begin{equation}
\label{eq:minimisation_problem_variable}
f^{\lambda, \mu}_{\mathrm{var}} := \argmin_{f \in \mathcal{F}} \ \widehat{\mathcal{R}}(f) + \lambda\Omega_0^2(f) + \mu\Omega_{\mathrm{var}}^r(f),
\end{equation}
with~$\lambda$ a fixed parameter and~$\mu$ a hyper-parameter to be estimated. When $r \geq 1$ \note{and the loss is convex}, we obtain a strongly-convex objective function, hence with a unique global minimiser. When $r<1$, which we use in practice, only a local minimiser can be reached.

\subsection{Hermite polynomials for feature learning}
\label{sec:feature}
We now turn to the feature learning setting described in Assumption~\ref{ass:feature}. The Hermite polynomials are particularly well-suited for feature learning, as they allow us to bridge the gap between variable selection and feature learning with only a minor modification of the previous penalties. This suitability is visible in some important properties which we now describe. First, the multivariate Hermite polynomials possess a rotation invariance property.

\begin{lemma}[\note{Rotational invariance property of Hermite polynomials}]
\label{lem:invariance_property}
For any~$x, x^\prime \in \mathbb{R}^d$, any~$k \in \mathbb{N}$ and any orthogonal matrix~$R \in O_d$, 
\begin{equation*}
\sum_{|\alpha| = k}  {H_\alpha(x) H_\alpha(x^\prime)} = \sum_{|\alpha| = k}  {H_\alpha(R x) H_\alpha(R x^\prime)}.
\end{equation*}
\end{lemma}
The proof of this lemma is available in Appendix~\ref{proof:invariance_property}. This property will be extremely useful to characterise the statistical behaviour of our methods, as discussed in Section~\ref{sec:statistical_consistency}. Another key property is that for any~$R \in O_d$, the family~$\big(H_\alpha(R\cdot)\big)_{\alpha \in \mathbb{N}^d}$ also forms a basis of~$L^2(q)$. Consequently, we can express any $f \in \mathcal{F}$ in this basis.

Moreover, we can characterise the derivatives of functions in~$L^2(q)$ as in Equation~\eqref{eq:derivatives_a}. Let~$f\in \mathcal{F}$ be written as~$f = \sum_{\alpha \in \mathbb{R}^d} \hat{f}(\alpha) H_\alpha$, then using Equation~\eqref{eq:partial_f}, we have the following expressions for the derivatives
\begin{equation}
\label{eq:derivatives_ab}
\int_{\mathbb{R}^d} \bigg(\frac{\partial f}{\partial x_a} \bigg) \bigg(\frac{\partial f}{\partial x_b} \bigg) q = \sum_{\alpha \in \mathbb{N}^d}\sqrt{(\alpha_a +1)}\sqrt{(\alpha_b+1)} \hat{f}(\alpha +e_a) \hat{f}(\alpha + e_b). 
\end{equation}

As before, we aim to enhance the regularisation using the sequence~$(c_k)_{k>0}$ For~$r \in (0, +\infty)$, we define
\begin{align}
   \nonumber \Omega_{\mathrm{feat}}(f)  &= \big( \trace \big( M_f^{r/2} \big)  \big)^{1/r} \\
\mbox{ with } (M_f)_{a,b} & =  \sum_{\alpha \in \mathbb{N}^d}  \frac{1}{c_{|\alpha|+1} }
 \sqrt{ \alpha_a+ 1} \sqrt{ \alpha_b+ 1}  \hat{f}(\alpha + e_a)  \hat{f}(\alpha + e_b), \ a,b \in [d]. \label{eq:m_f}
\end{align}

It is worth noting that~$M_f$ is a positive semi-definite matrix (see the proof of Lemma~\ref{lem:properties_of_omega_feat}). The penalty~$\Omega_{\rm feat}$  pushes the eigenvalues of of~$M_f$ towards~$0$, and since the rank of~$M_f$ is equal to the number of its non-zero eigenvalues, the penalty encourages the rank of~$M_f$ to be low. It is crucial that~$c_{|\alpha|}$ depends solely on~$|\alpha|$ and not on any other quantities \note{depending on}~$\alpha$ (e.g., ~$\max_{a \in [d]} \alpha_a$ for example). This property allows us to leverage the rotation invariance property described in Lemma~\ref{lem:invariance_property}, which is needed for our estimation algorithm in Section~\ref{sec:opti} and for obtaining statistical consistency results in Section~\ref{sec:statistical_consistency}.

Let us now examine some important properties of the proposed regularisation. \begin{lemma}[\note{Properties of the regularisation}]
\label{lem:properties_of_omega_feat}
For any~$f \in \mathcal{F}$, the following properties hold\begin{enumerate}
\item  Let~$R \in O_d$, if we define~$g = f(R \cdot)$, then~$M_f = R M_g R^\top$ and~$\Omega_{\rm feat}(f) = \Omega_{\rm feat}(g)$.
\item~$\Omega_{\rm var}(f) = \big( \trace \big( \Diag(M_f)^{r/2} \big)  \big)^{1/r}$.
\item If~$M_f$ is diagonal,~$\Omega_{\mathrm{feat}}(f) = \Omega_{\mathrm{var}}(f)$.
\item Let~$M_f = U D U^\top$ be the eigendecomposition of $M_f$, where~$U \in O_d$ and~$D$ is a diagonal matrix. If we define~$g=f(U\cdot)$, then~$M_g = D$ is diagonal and thus~$\Omega_{\mathrm{feat}}(f) = \Omega_{\mathrm{var}}(g)$.
\item Let~$M_f = U D U^\top$ be the eigendecomposition as above. If the rank of~$D$ is~$s$, then~$g = f(U\cdot)$ only depends on variables~$x_a$ where~$D_a>0$ and~$f=g(U^\top \cdot)$ only depends on~$s$ linear transformations of the original coordinates, namely of~$(U^\top x)_a$ for~$a$ such that~$D_a >0$.
\item If~$r=1, 
$\begin{equation*}
    \Omega_{\mathrm{feat}}(f) \geq  \inf_{R \in O_d} \Omega_{\mathrm{var}}(f(R \cdot)).
\end{equation*}
\end{enumerate}
\end{lemma}

\begin{proof}[\note{Proof of Lemma~\ref{lem:properties_of_omega_feat}}]
We proceed by proving each assertion separately.
\begin{enumerate}
\item   We have for $z \in \mathbb{R}^d$
\begin{align*}
z^\top M_f z & =  \sum_{a,b=1}^d
\sum_{\alpha \in \mathbb{N}^d}  \frac{1}{c_{|\alpha|+1} }
 z_a z_b \sqrt{ \alpha_a+ 1} \sqrt{ \alpha_b+ 1}  \hat{f}(\alpha + e_a)  \hat{f}(\alpha + e_b) \\
& =  \sum_{a,b=1}^d
\sum_{\alpha \in \mathbb{N}^d}  \frac{1}{c_{|\alpha|+1} }
 z_a z_b \Big\langle \frac{\partial f}{\partial x_a}  , H_\alpha \Big\rangle_{L^2(q)}
 \Big\langle \frac{\partial f}{\partial x_b}  , H_\alpha \Big\rangle_{L^2(q)}
\\
& = \sum_{\alpha \in \mathbb{N}^d}  \frac{1}{c_{|\alpha|+1} }
  \big\langle z^\top \nabla f   , H_\alpha \big\rangle_{L^2(q)}^2.
\end{align*}

This shows that $M_f$ is positive semi-definite, writing $\mathcal{N}(0, I_d)$ for the standard normal distribution on $\mathbb{R}^d$, we then have 
\begin{align*}
z^\top M_g z & =  \sum_{\alpha \in \mathbb{N}^d}   \frac{1}{c_{|\alpha|+1} }
\Big(  \mathbb{E}_{X \sim \mathcal{N}(0,I_d)}   \big( z^\top \nabla g (X)     H_\alpha (X) \big) \Big)^2 \\
& =  \sum_{\alpha \in \mathbb{N}^d}  \frac{1}{c_{|\alpha|+1} }
\Big(  \mathbb{E}_{X \sim \mathcal{N}(0,I_d)}   \big( z^\top  R^\top \nabla f (R X)     H_\alpha (X) \big) \Big)^2 \\
& \text{ as } \nabla g (X) =   R^\top \nabla f (R X)\\
& =  \sum_{\alpha \in \mathbb{N}^d}  \frac{1}{c_{|\alpha|+1} }
\Big(  \mathbb{E}_{X \sim \mathcal{N}(0,I_d)}   \big( z^\top  R^\top \nabla f (R X)     H_\alpha (R X) \big) \Big)^2\\
& \text{ by Lemma~\ref{lem:invariance_property},}\\
& =  \sum_{\alpha \in \mathbb{N}^d}   \frac{1}{c_{|\alpha|+1} }
\Big(  \mathbb{E}_{X \sim \mathcal{N}(0,I_d)}   \big( z^\top  R^\top \nabla f (X)     H_\alpha (X) \big) \Big)^2 \\
& \text{ by rotation invariance of the standard Gaussian}, \\
& =   z^\top R^\top M_f Rz  ,
\end{align*}
that is $M_g = R^\top M_f R$. The second assertion follows by the rotation invariance of the trace.
\item It suffices to see that for any $a \in [d]$
\begin{align*}
\Diag(M_f)_{a,a} &=  \sum_{\alpha \in \mathbb{N}^d} \frac{1}{c_{|\alpha|+1}} (\alpha_a+1)^2 \hat{f}(\alpha + e_a)^2   = \sum_{\alpha \in \mathbb{N}^d,  \alpha_a >0} \frac{1}{c_{|\alpha|}} \alpha_a \hat{f}(\alpha)^2,
    \end{align*} and therefore 
    \begin{equation*}
        \trace\big( \Diag(M_f)^{r/2} \big) = \sum_{a=1}^d \bigg( \sum_{\alpha \in \mathbb{N}^d} \alpha_a \frac{1}{c_{|\alpha|}} \hat{f}(\alpha)^2 \bigg)^{r/2} = \Omega_{\rm var}(f)^r.
    \end{equation*}
\item This is a direct consequence of the previous result, because of the definition of $\Omega_{\rm feat}$.
    \item By applying the first result, we find that $\Omega_{\rm feat}(f) = \Omega_{\rm feat}(g)$ and $M_g = D$. Then, using the third result, we conclude that $\Omega_{\rm var}(g) = \Omega_{\rm feat}(g)$. This establishes the desired result.
    \item  Consider the function $g = f(U\cdot)$. From the previous result, we know that $M_g = D$ is diagonal. According to the definition of $\Omega_{\rm var}$, we have $D_a = 0$ if and only if $g$ does not depend on variable $x_a$. Consequently, if the rank of $D$ is $s$, then $g$ only depends on $s$ variables, specifically those for which $D_a > 0$. As a result, we can conclude that $f = g(U^\top \cdot)$ depends solely on $(U^\top x)_a$ for $a$ such that $D_a > 0$.  
\item  Let us examine $\Omega_{\mathrm{feat}}$ and $\Omega_{\mathrm{var}}$ as follows
$$ 
  \Omega_{\mathrm{feat}}(f)  =  \big( \trace \big( M_f^{1/2} \big)  \big), \hspace*{1.5cm}
    \Omega_{\mathrm{var}}(f)  =  \big( \trace \big( \Diag(M_f)^{1/2} \big)  \big).
 $$
We can decompose $M_f$ as $M_f = U D U^\top$ using its eigendecomposition. If we define $g = f(U \cdot)$, then $M_g = D$ is diagonal, and we have $\Omega_{\mathrm{feat}}(f) = \Omega_{\mathrm{feat}}(g) = \Omega_{\mathrm{var}}(g)$. Consequently, we obtain the inequality
\begin{equation*}
    \Omega_{\mathrm{feat}}(f) \geq  \inf_{R \in O_d} \Omega_{\mathrm{var}}(f(R \cdot)).
\end{equation*}
\end{enumerate}
\end{proof}

The rotation invariance of $\Omega_{\mathrm{feat}}$ is crucial in the context of feature learning, as it ensures that the penalty is not biased towards specific directions. Similarly, $\Omega_0$ is also rotation invariant, as can be seen using Lemma~\ref{lem:invariance_property}. 

We observe that given a function $f$ and its associated matrix $M_f$, we can construct a function $g$ consisting of a rotation of the data and $f$ in such a way that the feature penalty on $f$ is equal to the variable selection penalty on~$g$. This highlights that the feature learning setting extends the variable selection problem by allowing data rotation. Furthermore, we can easily determine if $g$ depends only on a few variables, and therefore if $f$ depends only on a few linear transformations of the data, which aligns with our assumption for $f^*$. The last assertion of Lemma~\ref{lem:properties_of_omega_feat} will be useful to show that the proof of the consistency for the variable penalty easily extends to the feature learning setting, see Section~\ref{sec:statistical_consistency}.

With these considerations, we proceed to estimate~$f^*$ in the setting described by Assumption~\ref{ass:feature} by solving
\begin{align}
    \label{eq:minimisation_problem_feature}   f^{\lambda, \mu}_{\mathrm{feat}} &:= \argmin_{f \in \mathcal{F}} \  \widehat{\mathcal{R}}(f) + \lambda\Omega_0^2(f) + \mu \Omega_{\mathrm{feat}}^r(f), 
\end{align}
with~$\lambda$ a fixed parameter and~$\mu$ a hyper-parameter. We refer to this estimator as the \textbf{RegFeaL} (\textbf{R}egularised \textbf{F}eature \textbf{L}earning) estimator. \note{As for the relevant features or variables and dimension, we discuss their computation in Section~\ref{sec:variational}.}

\section{Estimator computation}
\label{sec:opti}
The computation of the solution for the optimisation problems delineated by \eqref{eq:minimisation_problem_variable} and \eqref{eq:minimisation_problem_feature} requires the employment of several strategic methodologies, which we will now discuss.

\subsection{Variational formulation}
\label{sec:variational}
We first use the following quadratic variational formulation, similar to the approach presented in \cite{livreviolet}. This formulation is necessary since it is not possible to directly optimise Equation~\eqref{eq:minimisation_problem_variable} and Equation~\eqref{eq:minimisation_problem_feature} due to the absence of closed-form solutions. Using other classical optimisation methods such as gradient-based methods would be less efficient as the overlapping group Lasso penalty \note{we propose} does not have efficient projection algorithms. Indeed, the variational formulation allows us to rewrite our optimisation problems as the minimisation over two variables of a specific quantity.  Subsequently, we can alternate the minimisation with respect to each variable, leading to rapid convergence in practice.

We first give the following Lemma which is adapted from \cite{jenatton2010structured}, which provides a variational formulation of sums of powers.
\begin{lemma}[\note{Variational formulation}]
\label{lem:variational}
Let~$r \in (0,2)$ and~$u \in \mathbb{R}_+^d$, then
\begin{equation*}
\|u\|_{r/2}^{r/2} = \bigg(\sum_{a=1}^d u_a^{r/2}\bigg) = \min_{\eta \in \mathbb{R}^d_+, \ \|\eta\|_{r/(2-r)} =1}     \sum_{a=1}^d \frac{u_a}{\eta_a},
\end{equation*}
with minimum attained at~$\eta, \forall a \in [d], \ \eta_a = u_a^{(2-r)/2}/\big(   \sum_{b=1}^d  u_b ^{r/2}\big)^{(2-r)/r}$.
\end{lemma}

Now, let us apply this approach to the penalty used for variable selection.
\begin{lemma}[\note{Variational formulation of variable selection penalty}]
\label{lemma:variational}
Let~$f= \sum_{\alpha \in \mathbb{N}^d} \hat{f}(\alpha) H_\alpha \in \mathcal{F}$ and~$r\in (0,2)$, then
\begin{align*}
    \Omega_{\mathrm{var}}^r(f) =&  \min_{\eta \in \mathbb{R}^d_+, \ \|\eta\|_{r/(2-r)} =1}  \sum_{a=1}^d \bigg(\sum_{\alpha \in (\mathbb{N}^d)^*} \alpha_a \frac{1}{c_{|\alpha|}} \hat{f}(\alpha)^2 \bigg)\eta_a^{-1} \\
    =& \min_{\eta \in \mathbb{R}^d_+, \ \|\eta\|_{r/(2-r)} =1}  \bigg(\sum_{\alpha \in (\mathbb{N}^d)^*} \alpha^\top \eta^{-1} \frac{1}{c_{|\alpha|}} \hat{f}(\alpha)^2 \bigg),
\end{align*}
where~$\eta^{-1} = (1/\eta_1, \ldots, 1/\eta_d)$ and where the minimum is reached for~$\eta$ such that
\begin{equation}
\label{eq:best_eta}
\forall a \in [d], \ \eta_a = \frac{\Big(\sum_{\alpha \in (\mathbb{N}^d)^*} \alpha_a \frac{1}{c_{|\alpha|}} \hat{f}(\alpha)^2 \Big)^{(2-r)/2}}{\Big(\sum_{b=1}^d \Big(\sum_{\alpha \in (\mathbb{N}^d)^*} \alpha_b \frac{1}{c_{|\alpha|}} \hat{f}(\alpha)^2 \Big)^{r/2}\Big)^{(2-r)/r}}.
\end{equation}
\end{lemma}

\begin{proof}[\note{Proof of Lemma~\ref{lemma:variational}}]
Recall $ \Omega_{\mathrm{var}} (f)  =  \big(\sum_{a = 1}^d \big( \sum_{\alpha \in (\mathbb{N}^d)^*}\frac{\alpha_a}{c_{|\alpha|}} \hat{f}(\alpha)^2 \big)^{r/2}\big)^{1/r}$ and use Lemma~\ref{lem:variational} with~$u_a =  \sum_{\alpha \in (\mathbb{N}^d)^*} \alpha_a \frac{1}{c_{|\alpha|}} \hat{f}(\alpha)^2$.
\end{proof}

We can then rewrite \eqref{eq:minimisation_problem_variable} as 
\begin{align}
\label{eq:final_variable_problem}
    f^{\lambda, \mu}_{\mathrm{var}}, \ \eta^{\lambda, \mu}_{\mathrm{var}} =& \argmin_{f \in \mathcal{F}, \>\> \eta \in \mathbb{R}^d_+}  \ \widehat{\mathcal{R}}(f) + \sum_{\alpha \in (\mathbb{N}^d)^*}  \frac{1}{c_{|\alpha|}} \hat{f}(\alpha)^2(\lambda + \mu\alpha^\top \eta^{-1}) \\
     &\mathrm{subject} \>\mathrm{to} \quad  f = \sum_{\alpha \in \mathbb{N}^d} \hat{f}(\alpha)H_\alpha, \quad 
      \|\eta\|_{r/(2-r)}=1. \nonumber
\end{align}

Recall that~$\Omega_{\rm var}(f) = \big(\sum_{a=1}^d \big(\sum_{\alpha \in (\mathbb{N}^d)^*} \alpha_a \frac{1}{c_{|\alpha|}} \hat{f}(\alpha)^2 \big)^{r/2}\big)^{1/2}$. Each term \newline $\big(\sum_{\alpha \in (\mathbb{N}^d)^*} \alpha_a \frac{1}{c_{|\alpha|}} \hat{f}(\alpha)^2 \big)^{r/2}$ quantifies the dependency of~$f$ on the variable~$x_a$. We then remark from the definition of~$\eta^{\lambda, \mu}_{\mathrm{var}}$ in Equation~\eqref{eq:best_eta}, that
\begin{equation*}
\forall a \in [d], \ \big(\eta^{\lambda, \mu}_{\mathrm{var}}\big)_a^{r/(2-r)} = \frac{\big(\sum_{\alpha \in (\mathbb{N}^d)^*} \alpha_a \frac{1}{c_{|\alpha|}} \hat{f}_{\rm var}^{\lambda, \mu}(\alpha)^2 \big)^{r/2}}{\sum_{b=1}^d \big(\sum_{\alpha \in (\mathbb{N}^d)^*} \alpha_b \frac{1}{c_{|\alpha|}} \hat{f}_{\rm var}^{\lambda, \mu}(\alpha)^2 \big)^{r/2}}.
\end{equation*}
Hence~$\big(\eta_{\rm var}^{\lambda, \mu}\big)_a$ represents the variation of~$f_{\mathrm{var}}^{\lambda, \mu}$ which is due to~$x_a$. We can use~$\eta_{\mathrm{var}}^{\lambda, \mu}$ to estimate the relevant underlying variables by using conventional techniques such as thresholding. Specifically, we can consider a variable~$x_a$ to be relevant only if~$\eta_a$ is above some predetermined threshold, i.e.~$\hat{S} := \{ a \in [d], \ \big(\eta_{\rm var}^{\lambda, \mu}\big)_a > t \}$ for some~$t>0$.

We can proceed in a similar manner for the feature learning setting. 
\begin{lemma}[\note{Variational formulation of feature learning penalty}]
Let~$f \in \mathcal{F}$,~$M_f$ from Equation~\eqref{eq:m_f}, with~$M_f = U D U^\top$ its eigendecomposition and~$r\in (0,2)$, then
\begin{align*}
 \Omega_{\mathrm{feat}}^r(f) = \min_{\Lambda \in \mathbb{R}^{d \times d}}  \trace & \big( \Lambda^{-1} M_f \big) \\
  \mathrm{subject} \>\mathrm{to} \ &  \Lambda = R\Diag{(\eta)}R^\top \\
 & R \in O_d, \ \eta  \in \mathbb{R}^d_+, \ \|\eta\|_{r/(2-r)} =1,
\end{align*}
where the minimum is attained for 
\begin{align}
\label{eq:best_lambda}
   \Lambda &= U \Diag{(\eta)}  U^\top  \\
   \nonumber \forall a \in [d], \ \eta_a &= \frac{D_a^{(2-r)/2}}{(\sum_{b=1}^d D_b^{r/2})^{(2-r)/r}}.
\end{align}
\end{lemma}
This allows us to rewrite Equation~\eqref{eq:minimisation_problem_feature} as
\begin{align*}
    f^{\lambda, \mu}_{\mathrm{feat}}, \ \Lambda^{\lambda, \mu}_{\mathrm{feat}} = \argmin_{f \in \mathcal{F}, \ \Lambda \in \mathbb{R}^{d \times d}} &\ \widehat{\mathcal{R}}(f) + \lambda \Omega_0^2(f) + \mu \trace{\big(\Lambda^{-1}M_f\big)} \\
    \mathrm{subject} \>\mathrm{to} & \ \Lambda = R \Diag(\eta)R^\top  \\
    & R \in O_d, \ \eta \in \mathbb{R}^d_+, \  \|\eta\|_{r/(2-r)}=1.
\end{align*}

Moreover, with~$\Lambda = R\Diag(\eta) R^\top$ as above, if we write~$f~$ in the rotated basis as~$f =  \sum_{\alpha \in \mathbb{N}^d} \hat{f}(\alpha)H_\alpha(R^\top \cdot)$, and~$g=f(R\cdot) = \sum_{\alpha \in \mathbb{N}^d} \hat{f}(\alpha)H_\alpha$, we have~$M_f = R M_g R^\top$ (Lemma~\ref{lem:properties_of_omega_feat}). Therefore  
\begin{align*}
\trace{(\Lambda^{-1} M_f)} = \trace{\big(\Diag{(\eta^{-1})}M_g\big)}  &= \sum_{a=1}^d \eta_a^{-1} \sum_{\alpha \in (\mathbb{N}^d)^*} \frac{\alpha_a}{c_{|\alpha|}} \hat{f}(\alpha)^2 \\
& =  \sum_{\alpha \in (\mathbb{N}^d)^*}  \frac{1}{c_{|\alpha|}} \hat{f}(\alpha)^2\alpha^\top \eta^{-1}.
\end{align*}
We can then rewrite Equation~\eqref{eq:final_feature_problem} as
\begin{align}
\label{eq:final_feature_problem}
f^{\lambda, \mu}_{\mathrm{feat}}, \ \Lambda^{\lambda, \mu}_{\mathrm{feat}} = \argmin_{f \in \mathcal{F}, \ \Lambda \in \mathbb{R}^{d \times d}} & \  \widehat{\mathcal{R}}(f) + \sum_{\alpha \in (\mathbb{N}^d)^*}  \frac{1}{c_{|\alpha|}} \hat{f}(\alpha)^2(\lambda + \mu\alpha^\top \eta^{-1}) \\
\mathrm{subject} \>\mathrm{to} &  \  \Lambda = R \Diag(\eta)R^\top \nonumber  \\
& \ R \in O_d, \ \eta \in \mathbb{R}^d_+,  \ \|\eta\|_{r/(2-r)}=1 \nonumber \\
& \ f=\sum_{\alpha \in \mathbb{N}^d }\hat{f}(\alpha) H_\alpha(R^\top \cdot). \nonumber
\end{align}

We see then that the feature learning problem can be viewed as an extension of the variable selection problem, where we additionally optimise over any possible data rotation. Conversely, the variable selection problem can be seen as a particular case of the feature learning problem, where the rotation matrix~$R$ is fixed to the identity matrix. 

To estimate the dimension of the underlying feature space $P$ and the features themselves, we use the eigendecomposition of $\Lambda^{\lambda, \mu}_{\mathrm{feat}} = (R^{\lambda, \mu}_{\mathrm{feat}})^\top \Diag(\eta_{\mathrm{feat}}^{\lambda, \mu})R^{\lambda, \mu}_{\mathrm{feat}}$. By using the columns of $R^{\lambda, \mu}_{\mathrm{feat}}$ corresponding to the selected features, denoted as $\hat{S}:= \{a \in [d] \, | \, \left(\eta_{\mathrm{feat}}^{\lambda, \mu}\right)_a > t\}$ for some threshold $t > 0$, we construct our feature estimator $\hat{P}$, i.e., $\hat{P} := \big(R_{\mathrm{feat}}^{\lambda, \mu}\big)_{\hat{S}}$. We see that by employing alternating minimisation, we are able to simultaneously learn the regression function and the underlying features.

\subsection{Optimisation procedure}
We now discuss how to solve the optimisation problem using alternative minimisation, drawing on techniques described in \cite{livreviolet}.  In the following discussion, we will focus on the feature learning setting. However, it is important to note that by simply fixing $R=I_d$ in each equation, we can easily revert back to the variable selection case.

To solve Equation~\eqref{eq:final_feature_problem}, we have observed that when the function~$f$ is fixed, the optimal~$\Lambda$ can be determined using Equation~\eqref{eq:best_lambda}, which involves the matrix~$M_f$.\footnote{If~$f= \sum_{\alpha} \hat{f}(\alpha) H_\alpha(R^\top \cdot)$, to compute~$M_f$, we can remark that with~$g = f(R\cdot) = \sum_{\alpha} \hat{f}(\alpha) H_\alpha$, we have the usual formula for~$M_g$ from Equation~\eqref{eq:m_f} and~$M_f= R M_g R^\top$.} 

When~$\Lambda$ is fixed, we seek to solve the optimisation problem
\begin{align*}
\argmin_{f \in \mathcal{F}} & \ \widehat{\mathcal{R}}(f) + \sum_{\alpha \in (\mathbb{N}^d)^*}  \frac{1}{c_{|\alpha|}} \hat{f}(\alpha)^2(\lambda + \mu\alpha^\top \eta^{-1}) \\
\mathrm{subject} \>\mathrm{to} & \ f=\sum_{\alpha \in \mathbb{N}^d }\hat{f}(\alpha) H_\alpha(R^\top \cdot),
\end{align*}
where~$\Lambda = R \Diag(\eta) R^\top$. However, this can only be solved if~$\widehat{\mathcal{R}}$ is known, i.e., for some chosen loss function~$\ell$. Until the end of Section~\ref{sec:opti}, we consider the quadratic loss which is commonly used in regression problems and allows for closed-form solutions. Otherwise, iterative optimisation algorithms \note{need} to be employed. The problem is then 
\begin{align}
\label{eq:explicit_kernel_min}
\argmin_{f \in \mathcal{F}} & \ \frac{1}{n}\sum_{i=1}^n \big( y^{(i)} - f(x^{(i)}) \big)^2 + \sum_{\alpha \in (\mathbb{N}^d)^*}  \frac{1}{c_{|\alpha|}} \hat{f}(\alpha)^2(\lambda + \mu\alpha^\top \eta^{-1}) \\
\mathrm{subject} \>\mathrm{to} & \  f=\sum_{\alpha \in \mathbb{N}^d }\hat{f}(\alpha) H_\alpha(R^\top \cdot). \nonumber
\end{align}
If we write for any $x, x^\prime \in \mathbb{R}^d$
\begin{equation}
\label{eq:kernel}
k_\Lambda(x,x^\prime) = \sum_{\alpha \in  (\mathbb{N}^d)^*}    \frac{c_{|\alpha|} H_\alpha(R^\top x) H_\alpha(R^\top x^\prime)}{   \lambda + \mu\alpha^\top \eta^{-1}},
\end{equation}
the function~$k_\Lambda$ verifies all properties required to be a reproducing kernel \cite{aronszajn50reproducing}. The condition for a function to be a reproducing kernel is that it is symmetric and that the associated kernel matrix is positive definite for any set of points. Specifically, for any $n \in \mathbb{N}$ and $x^{(1)}, \ldots, x^{(n)}$, the matrix $K_\Lambda = (k_\Lambda(x^{(i)}, x^{(j)}))_{i,j \in [n]}$ must be positive definite (where $\lambda > 0$ is useful in this context). We can then apply the theory of reproducing kernel Hilbert spaces (RKHS).  In this case, $k_\Lambda$ serves as the reproducing kernel for the space $\mathcal{F}$, with associated norm $\| \cdot \|_\Lambda$, given by
\begin{equation*}
    \| f \|^2_\Lambda =  \sum_{\alpha \in (\mathbb{N}^d)^*}  \frac{1}{c_{|\alpha|}} \hat{f}(\alpha)^2(\lambda + \mu\alpha^\top \eta^{-1})
\end{equation*}
(note that $\hat{f}$ depends on $\Lambda$ through $R$).
We can interpret the problem as a standard kernel ridge regression, which we refer to as the ``kernel point of view.'' By applying the representer theorem \cite{aronszajn50reproducing}, we know that the solution to Equation~\eqref{eq:explicit_kernel_min} takes the form
\begin{align*}
    f &= \sum_{i=1}^n \delta^{\Lambda}_i k_\Lambda(x^{(i)}, \cdot) + \delta_0^{\Lambda} ,
\end{align*}
where~$\delta^{\Lambda}$ and~$\delta_0^{\Lambda}$ can be obtained in closed form using ~$Y=(y^{(1)}, \ldots, y^{(n)})^\top$ and~$K=(k_\Lambda(x^{(i)}, x^{(j)}))_{i,j \in [n]}$ as the minimisers of 
\begin{equation}
\label{eq:update_delta}
    \delta^{\Lambda}, \ \delta_0^{\Lambda} = \argmin_{\delta \in \mathbb{R}^n, \ \delta_0 \in \mathbb{R}} \frac{1}{n} \| Y - K_\Lambda \delta - \delta_0 \mathds{1}\|_2^2 + \delta^\top K_\Lambda \delta.
\end{equation}

It is worth noting that the shape of the kernel defined in Equation~\eqref{eq:kernel} implies that features corresponding to $\alpha \in \mathbb{N}^d$ with large values of $\alpha^\top \eta^{-1}$ are penalised more. If $\eta_a$ is small, indicating that it has been pushed down in the previous optimisation steps, it suggests that variable $x_a$ or the direction $(R^\top x)_a$ may not be particularly useful for prediction. In such cases, for these variables/directions to be retained, they would need to contribute significantly more to the fit compared to others.

Furthermore, we observe that the parameter $\lambda$ serves the purpose of ensuring numerical stability when solving linear systems, particularly when $\alpha^\top \eta^{-1}$ can be null. We recommend setting $\lambda$ to a significantly smaller value than $\mu$ to achieve this desired stability (e.g, $\lambda = 10^{-8} / d^{(2-r)/r}$ in our experiments). In fact, it is possible to fix $\lambda$ as a predetermined value, eliminating the need for it to be treated as a hyper-parameter.

\subsection{Sampling approximation of the kernel}
\label{sec:sampling}
We remark that the kernel described in Equation~\eqref{eq:kernel} is defined as an infinite sum, which means it is not computable in practice.  To overcome this challenge, we adopt an approximation approach using sampling.

Let us define~$C(\eta) = \sum_{\alpha \in (\mathbb{N}^d)^*} \frac{c_{|\alpha|}}{\lambda + \mu \alpha^\top \eta^{-1}}$. By defining~$ h(\alpha) = \frac{1}{C(\eta)} \frac{c_{|\alpha|}}{\lambda + \mu \alpha^\top \eta^{-1}}$, for all $\alpha \in (\mathbb{N}^d)^*$ we obtain a probability distribution on $(\mathbb{N}^d)^*$. Consequently, we can express the kernel $k_\Lambda(x, x^\prime)$ as $C(\eta) \mathbb{E}_{\alpha \sim h} \big( H_\alpha(R^\top x) H_\alpha(R^\top x^\prime) \big)$. 

Sampling from the distribution $h$ can be challenging, particularly in high-dimensional settings. Therefore, we employ an importance sampling technique. For the first choice~$c_{|\alpha|} = \mathds{1}_{|\alpha| \leq M}$,  the kernel $k_\Lambda(x,x^\prime)$ can be expressed as
\begin{align*}
    k_\Lambda(x,x^\prime) &= \sum_{\alpha \in (\mathbb{N}^d)^*, |\alpha| \leq M} \frac{H_\alpha(R^\top x) H_\alpha(R^\top x^\prime)}{\lambda + \mu\alpha^\top\eta^{-1}} \\
    & =  \binom{M + d }{d}  \mathbb{E}_{\alpha \sim \mathcal{U}\{ |\alpha| \leq M \}}\bigg( \frac{H_\alpha(R^\top x)H_\alpha(R^\top x^\prime)}{\lambda + \mu \alpha^\top \eta^{-1}} \bigg),
\end{align*}
where $\mathcal{U}\{ |\alpha| \leq M \}$ is the uniform distribution over $\{ \alpha \in (\mathbb{N}^d)^*, \ |\alpha| \leq M \}$. Sampling from this uniform distribution can be achieved by selecting a subset~$B$ of size $d$ uniformly from the set $[M+d]$, sorting the subset into $B_1 < \cdots < B_d$, setting $B_0=0$, and using the differences between consecutive values to construct $\alpha$. Specifically, for each $a \in [d]$, we set $\alpha_a = B_a - B_{a-1} -1$. If the resulting~$\alpha$ is the null tuple, it is rejected, and the sampling process is repeated.

For the choice~$c_{|\alpha|} = \rho^{|\alpha|}$ the kernel is
\begin{align*}
    k_\Lambda(x, x^\prime) &= \sum_{\alpha \in (\mathbb{N}^d)^*} \frac{\rho^{|\alpha|}}{\lambda + \mu \alpha^\top \eta^{-1}} H_\alpha(R^\top x) H_\alpha(R^\top x^\prime).
\end{align*}

We have developed a methodology called ``group sampling'' that addresses the challenges of sampling from the distribution $h$. To initialise the sampling, we set all components of $\eta$ to be equal. This choice ensures unbiasedness among the possible directions while satisfying the constraint $\|\eta\|_{r/(2-r)} = 1$. As a result, the kernel takes the form
\begin{align*}
    k_\Lambda(x, x^\prime) = \sum_{\alpha \in (\mathbb{N}^d)^*} \frac{\rho^{|\alpha|}}{\lambda + \mu |\alpha|d^{(2-r)/r}} H_\alpha(R^\top x) H_\alpha(R^\top x^\prime).
\end{align*}

We can directly sample from the distribution proportional to~$\frac{\rho^{|\alpha|}}{\lambda + \mu |\alpha|d^{(2-r)/r}}$. The sampling process involves two steps. First, we sample an integer $k$ from the distribution
\begin{equation*}
    k \sim \binom{k+d -1}{d-1} \frac{\rho^{k}}{\lambda + \mu d^{(2-r)/r} k }.
\end{equation*}
To perform this sampling, we can precompute a table of probabilities for different values of $k$ up to a chosen maximum value (e.g., 40). We then normalise these probabilities and use them to sample the value of $k$. Once we have obtained $k$, it represents the cardinality of $\alpha$. In the second step, we sample $\alpha$ uniformly from the set ${ \alpha \in (\mathbb{N}^d)^*, \ |\alpha| = k}$. This sampling procedure is exact, except for the controlled approximation introduced by the choice of the maximum value. 

We can develop an importance sampling scheme for the other optimisation steps when the components of $\eta$ are not equal. Here are the steps.
\begin{enumerate}
    \item Sort the components of $\eta$ in ascending order and find the largest gap between consecutive values. Divide the set $[d]$ into two groups: Group 1, containing the components above the top of the gap, with size $d_1$, and Group 2, containing the remaining components, with size $d_2$.
    \item Define $\Tilde{\eta}_1$ as the minimum value among the components in Group 1, and $\Tilde{\eta}_2$ as the maximum value among the components in Group 2.
    \item Sample $k_1$ and $k_2$ from the distribution
\begin{equation*}
k_1, k_2 \sim \binom{k_1+d_1 -1}{d_2-1}\binom{k_2+d_2 -1}{d_2-1} \frac{\rho^{k_1 + k_2}}{\lambda + \mu \left(\frac{k_1}{\Tilde{\eta}_1} + \frac{k_2}{\Tilde{\eta}_2}\right)},
\end{equation*}
where $k_1$ and $k_2$ represent $|\alpha^{(1)}|$ and $|\alpha^{(2)}|$ respectively, and $\alpha^{(1)}$ corresponds to the components in Group 1.
\item 
Sample $\alpha^{(1)}$ uniformly from the set ${ \alpha \in (\mathbb{N}^{d_1})^, \ |\alpha| = k_1}$, and sample $\alpha^{(2)}$ uniformly from the set ${ \alpha \in (\mathbb{N}^{d_2})^, \ |\alpha| = k_2}$.
\item This yields
\begin{align*}
    k_\Lambda(x, x^\prime) &= \sum_{\alpha \in (\mathbb{N}^d)^*}  \frac{C(\Tilde{\eta})}{C(\Tilde{\eta})}\frac{\rho^{|\alpha|}}{\lambda + \mu \alpha^\top \Tilde{\eta}^{-1}} \frac{\lambda + \mu \alpha^\top \Tilde{\eta}^{-1} }{\lambda + \mu \alpha^\top \eta^{-1}} H_\alpha(R^\top x) H_\alpha(R^\top x^\prime) \\
    & = \mathbb{E}_ {\alpha \sim \mathrm{Group \  sampling}}  \bigg(C(\Tilde{\eta})\frac{\lambda + \mu \alpha^\top \Tilde{\eta}^{-1} }{\lambda + \mu \alpha^\top \eta^{-1}} H_\alpha(R^\top x) H_\alpha(R^\top x^\prime) \bigg),
\end{align*}
with~$C(\Tilde{\eta})$ a normalising constant. 
\end{enumerate}
By using this importance sampling scheme, we can approximate the desired distribution accurately, even when the components of $\eta$ are not equal.

We observe that with the group sampling approach, the distribution of $\alpha$ is influenced by $\eta$ through the grouping process, as well as through the values of $\Tilde{\eta}_1$ and $\Tilde{\eta}_2$. As the optimisation progresses, the sampled tuples exhibit specific patterns: in directions that are deemed unimportant (corresponding to small values of $\eta_a$), $\alpha_a$ tends to be close to zero, while in directions that are considered important (corresponding to large $\eta_a$), $\alpha_a$ is more widely distributed.\footnote{It is worth noting that using the geometric distribution independently on each dimension of $\alpha$ would have been a simpler approach. However, this method becomes highly inefficient as the dimensionality increases, since it would involve sampling numerous $\alpha$ tuples with low importance weights (as determined by $\frac{H_\alpha(R^\top x) H_\alpha(R^\top x^\prime)}{\lambda + \mu \alpha^\top \eta^{-1}}$) due to their alignment with directions where $\eta$ is very small (i.e., $\alpha^\top \eta^{-1}$ is small).}

No matter the sampling scheme, we sample~$\alpha^{(1)}, \ldots, \alpha^{(m)}$ from some distribution with importance weight~$w(\alpha)$, yielding 
\begin{equation*}
    k_\Lambda(x, x^\prime) \approx \sum_{j=1}^m w(\alpha^{(i)}) H_{\alpha^{(j)}}(R^\top x) H_{\alpha^{(j)}}(R^\top x^\prime).
\end{equation*}
We use this formula to compute the kernel matrix~$K_\Lambda = \big(k_\Lambda(x^{(i)}, x^{(j)})\big)_{i,j \in [n]}$. Instead of approximating the matrix~$K_\Lambda$ to use in Equation~\eqref{eq:update_delta}, we can also consider the equivalent explicit ``feature point of view'' by writing~$f$ in the form 
\begin{equation*} 
    f=  \sum_{j=1}^{m} \theta_j w(\alpha^{(j)})H_{\alpha^{(j)}}(R^\top \cdot) + \theta_0 H_{0},
\end{equation*}
where
\begin{equation}
\label{eq:update_theta}
    \theta^{\Lambda}, \ \theta^{\Lambda}_0 = \argmin_{\theta \in \mathbb{R}^m, \ \theta_0 \in \mathbb{R}} \frac{1}{n} \|Y - \Phi \theta - \theta_0 \mathds{1}\|^2_2 + \|\theta\|_2^2 ,
\end{equation}
with~$\Phi \in \mathbb{R}^{n \times m}$ the matrix filled with~$w(\alpha^{(j)})H_{\alpha^{(j)}}( R^\top x^{(i)})$. This is computationally advantageous when~$n > m$. Otherwise, we use the kernel point of view. In both cases, we can use~$(\theta, \theta_0)$ or~$(\delta, \delta_0)$ to rewrite~$f$ as~$\sum_{\alpha \in \mathbb{N}^d} \hat{f}(\alpha) H_\alpha(R^\top \cdot)$. We remark that~$\hat{f}(\alpha) = 0$ when~$\alpha$ has not been sampled. 

The pseudo-code for the \textbf{RegFeaL} method is provided in Algorithm~\ref{alg:reg_feal}.
\begin{algorithm}[h]
\caption{RegFeaL pseudocode}\label{alg:reg_feal}
\For{$i \in [n_{\rm iter}]~$}{
\eIf{$i=0$}{$\eta \gets \mathds{1}/d^{(2-r)/r}$\;
   ~$R \gets I_d$\;
}{\eIf{\rm feature learning}{Update~$R$ and~$\eta$ as in Equation~\eqref{eq:best_lambda}\;}{Update~$\eta$ as in Equation~\eqref{eq:best_eta}\;}
}    
Sample~$\alpha^{(1)}, \ldots, \alpha^{(m)}$ using group sampling as in Section~\ref{sec:sampling} with~$\eta$ \; Compute importance weights~$w(\alpha^{(1)}), \ldots, w(\alpha^{(m)})$\;
Compute Hermite features~$\Phi \in \mathbb{R}^{n \times m}, \Phi_{i,j} =  w(\alpha^{(j)})H_{\alpha^{(j)}}( R^\top x^{(i)})$  \;
\eIf{$n>m$}{Update~$\theta$ as in Equation~\eqref{eq:update_theta}\;}{Update~$\delta$ as in Equation~\eqref{eq:update_delta}\;}
}
\end{algorithm}

In terms of numerical complexity, each iteration has a cost of 
\begin{equation*}
\mathcal{O}\bigg(\underset{\text{Hermite features}}{nm^\prime d + nd^2} + \underset{M_f  \text{ and its eigendecomposition}}{d^2(m^\prime)^2 +d^3} + \underset{\text{Sampling}}{m d} + \underset{\text{Computing } \theta \text{ or } \delta}{nm^\prime\max(n,m^\prime)} \bigg),
\end{equation*}
where $m^\prime$ is the number of unique tuples sampled (which is necessarily smaller than $m$, and can be much smaller when $\eta$ is sparse). The parameter $m$ can be chosen to achieve a balance between computational cost and performance, but selecting an excessively small value for $m$ may adversely affect performance. In practice, the number of iterations required for convergence is typically very small (less than $10$), as demonstrated in Section~\ref{sec:num}. Additionally, it is worth noting that the computation cost of $\delta$ in the feature point of view could be reduced through the use of the Nyström approximation \cite{nystrom_rudi}.

\section{Statistical properties}
\label{sec:statistical_consistency}
We now consider the statistical properties of \textbf{RegFeaL}. We always take~$r=1$ and we do not consider the approximation due to the computation of the estimators in this section.  \note{Our goal is to provide a high-probability bound on the expected risk of \textbf{RegFeaL} to gain insights into its generalisation properties under minimal assumptions to obtain a very general result. We do not consider the consistency of the e.d.r. space estimation, as this usually requires much stronger assumptions, such as the linearity condition, the gradient along the relevant directions to be large enough in norm, or constraint on the loss to be the square loss, for example \cite{cook_review_of_li, review}.}

We leverage the results presented in \cite{francis_book}, which provide bounds on the maximum difference between empirical and expected risk, in terms of the expectation over the class of functions with bounded norm. These bounds are expressed in terms of the Rademacher complexity of the set \note{$\{ f \in \mathcal{F}, \ \Omega(f) \leq D\}$}, where $D > 0$ is a fixed bound. By employing these results, we can obtain a probabilistic bound on the constrained estimator and apply McDiarmid's inequality \cite{boucheron2013concentration} to establish a result in probability. Ultimately, Theorem~\ref{theo:stat_general} provides a probabilistic bound for the \textbf{RegFeaL} estimator, leveraging the aforementioned results as well as the optimality conditions satisfied by the estimator.

\subsection{Setup}
\label{sec:theo_setup}
We start by making assumptions about the data used to train the model.
\begin{assumption}[Data]
 \label{ass:sampling}   
$\mathcal{D}= (x^{(i)},y^{(i)})_{i \in [n]}$ is a  set of i.i.d data, with~$(X, Y)$ a pair of random variables such that~$\forall i \in [n], \ (x^{(i)}, y^{(i)}) \sim (X,Y)$.
\end{assumption}
\note{Notice that we do not make strong assumptions on the distribution of the data, such as independence of the covariates or constraint to be elliptically contoured, nor do we require is to be known a priori.}

Let us introduce some definitions. Let $(c_k)_{k >0}$ be a non-null sequence of positive reals. We define the function space $\mathcal{F}$ as $\Vect \big( \{H_0\} \cup \{H_\alpha,  \text{ for } \alpha \in (\mathbb{N}^d)^* \text{such that } c_{|\alpha|>0} \} \big)$. Let $\ell$ be a loss function on~$\mathbb{R} \times \mathbb{R}$, and let the expected risk~$\mathcal{R}$ and the empirical risk~$\widehat{\mathcal{R}}$ be
\begin{align*}
\mathcal{R}(f) = \mathbb{E}_{X,Y}\big(\ell(Y, f(X))\big) \ \text{ and } \ \widehat{\mathcal{R}}(f) = \frac{1}{n}\sum_{i=1}^n \ell(y^{(i)}, f(x^{(i)})).
\end{align*} 

We define the functional norm $\Omega(f)$ for any $f \in \mathcal{F}$ as $\Omega(f) := \Omega_{\rm feat}(f) + |\hat{f}(0)|$ or $\Omega(f) := \Omega_{\rm var}(f) + |\hat{f}(0)|$, where $\hat{f}(0)$ represents the constant coefficient of $f$. It is important to note that the constraint on the constant coefficient is not necessary in practice, but we include it for the purpose of theoretical analysis (we could also add a small weight on $|\hat{f}(0)|$ to this effect). We define the regularised empirical risk $\widehat{\mathcal{R}}\mu(f)$ for $\mu > 0$ as follows
\begin{equation*}
\widehat{\mathcal{R}}\mu(f) = \frac{1}{n} \sum_{i=1}^n \ell(y^{(i)}, f(x^{(i)})) + \mu\Omega(f).
\end{equation*}
We denote our estimator as $f^\mu := \argmin_{f \in \mathcal{F}} \widehat{\mathcal{R}}_\mu$. In order to establish theoretical results, we will rely on the following assumptions. 
\begin{assumption}[Problem assumptions]
\label{ass:problem}
\begin{enumerate}
\item The true regression function $f^* := \argmin_{f \in \mathcal{F}} \mathcal{R}(f)$ exists.
\item For some~$D>0$, the loss function $\ell$ is $G$-Lipschitz continuous in its second argument for any value of its first argument, i.e.,~$\forall y \in \mathcal{Y}, \ \forall x, x^\prime \in \mathcal{X}, \ \forall f \in \mathcal{F}$ such that $ \Omega(f) \leq D, \ |\ell(y, f(x)) - \ell(y, f(x^\prime))| \leq G\cdot |f(x) - f(x^\prime) |$.
\item For some $D > 0$,   $\ell_\infty := \sup_{(x, y) \in \mathcal{X} \times \mathcal{Y}, f \in \mathcal{F}, \Omega(f) \leq D} \ell(y, f(x))$ is finite.
\item The loss $\ell$ is convex on~$\mathbb{R} \times \mathbb{R}$ .
\end{enumerate}
\end{assumption}
For our main result, we will use~$D = 2\Omega(f^*)$. These assumptions are commonly used in the analysis of nonparametric regression~\cite{gyorfi2002distribution}. Many commonly used loss functions in regression problems, such as the quadratic loss, absolute mean error, Huber loss, or logistic loss, are convex.  The Lipschitz continuity condition holds for all of these losses, except for the quadratic loss, which we handle separately, for example by exploiting the boundedness of the data. If the data is bounded (i.e., $\mathcal{X}\times \mathcal{Y}$ is bounded in $\mathbb{R}^d \times \mathbb{R}$), then $\sup_{x \in \mathcal{X}, f \in \mathcal{F},  \Omega(f) \leq D} |f(x)|$ is bounded for any $D>0$.\footnote{This can be seen for $\Omega_{\rm var}$ by noticing that $\Omega(f)$ can be written as $|\hat{f}(0)| + \sum_{a=1}^d \Theta_a(f) \geq \big(|\hat{f}(0)|^2 + \sum_{a=1}^d \Theta_a(f)^2\big)^{1/2}$, with the latter being an RKHS norm with reproducing kernel $k(x, x^\prime) = 1+ \sum_{\alpha \in (\mathbb{N}^d)^*} \frac{c_{|\alpha|}}{|\alpha|} H_\alpha(x) H_\alpha(x^\prime)$. It follows that $f(x) =  \langle f, k(X, \cdot) \rangle \leq \hat{f}(0) + \Omega(f) \sqrt{k(x,x)}$ which is bounded if $x$ is bounded.} We can then use the convexity of the loss $\ell$ and boundedness of $\mathcal{Y}$ to justify that $\ell_\infty$ is well-defined.  For the quadratic loss, in this setting, it satisfies Assumption \ref{ass:problem}.2 because $(y - f(x))^2 - (y - f(x^\prime))^2 = (f(x^\prime) - f(x))(y - f(x) + y - f(x^\prime))$, and we can then take$G := \sup_{(x, y) \in \mathcal{X} \times \mathcal{Y}, f \in \mathcal{F}, \Omega(f) \leq D} |y - f(x) + y - f(x^\prime)|$.

\subsection{Rademacher complexity}
First, we apply the Lipschitz continuity assumption to bound the supremum over a set of functions of the difference between the empirical risk and expected risk, in expectation over the dataset.
\begin{lemma}[\note{Use of Gaussian complexity}]
\label{lemma:expectation_sups}
Let~$\mathcal{G}$ be any set of functions, then under Assumption~\ref{ass:sampling}, and Assumption~\ref{ass:problem}.2,
    \begin{equation*}
        \mathbb{E}_{\mathcal{D}}\bigg(\sup_{f \in \mathcal{G}} \big( \mathcal{R}(f) - \widehat{\mathcal{R}}(f) \big) + \sup_{f \in \mathcal{G}} \big( \widehat{\mathcal{R}}(f) - \mathcal{R}(f) \big) \bigg)\leq 4\sqrt{\frac{\pi}{2}}G \cdot G_n(\mathcal{G}) ,
    \end{equation*}
where 
\begin{equation*}
     G_n(\mathcal{G}) := \mathbb{E}_{\mathcal{D}, \varepsilon \sim \mathcal{N}(0,I_n) }\bigg( \sup_{f \in \mathcal{G}} \frac{1}{n}\sum_{i=1}^n \varepsilon_i f(x^{(i)}) \bigg)
\end{equation*}
is the Gaussian complexity of the set~$\mathcal{G}$~\cite{bartlett2002rademacher}.
\end{lemma}

See Appendix~\ref{proof:expectation_sups} for the proof, which we include for the sake of completeness. This is a close variation of the work presented in \cite{francis_book}. We now need to bound the Gaussian complexity, when we consider subsets of the working space~$\mathcal{F}$ with bounded norm. 
\begin{lemma}[\note{Bound on Gaussian complexity}]
\label{lemma:rademacher_Gaussian}
Let~$D >0$, with~$\mathcal{G} := \{f \in \mathcal{F},\ \Omega(f) \leq D\}$ with~$\Omega$ defined as in Section~\ref{sec:theo_setup}, under Assumption~\ref{ass:sampling}, we have
\begin{equation*}
    G_n(\mathcal{G}) \leq D \cdot \sqrt{\frac{1}{n}  \bigg( 1 + \sum_{\alpha \in (\mathbb{N}^d)^*} \frac{c_{|\alpha|}}{|\alpha|}\mathbb{E}_{X}\big(H_\alpha(X)^2\big)} \bigg).
\end{equation*} 
\end{lemma}
We remark that the result depends heavily on the data distribution through the expectations~$\mathbb{E}_X(H_\alpha(X)^2)$ and the design of the norm through~$(c_k)_{k>0}$. We discuss these in more details in Section~\ref{sec:dependence}.

\begin{proof}[\note{Proof of Lemma~\ref{lemma:rademacher_Gaussian}}]
We first consider the norm~$\Omega_{\mathrm{var}}$. Let~$f \in \mathcal{G}$, we have
\begin{equation*}
    \frac{1}{n} \sum_{i=1}^n \varepsilon_i f(x^{(i)}) =  \sum_{\alpha \in \mathbb{N}^d} \hat{f}(\alpha) \bigg( \frac{1}{n} \sum_{i=1}^n \varepsilon_i H_\alpha(x^{(i)}) \bigg) =  \sum_{\alpha \in \mathbb{N}^d} \hat{f}(\alpha) \hat{\xi}(\alpha),
\end{equation*}
with~$\xi$ an infinite vector indexed by~$(\mathbb{N}^d),  \hat{\xi}(\alpha) = \frac{1}{n} \sum_{i=1}^n \varepsilon_i H_\alpha(x^{(i)})~$. Therefore 
\begin{equation*}
    \sup_{f \in \mathcal{G}} \frac{1}{n} \sum_{i=1}^n \varepsilon_i f(x^{(i)}) = \sup_{f \in \mathcal{G}} \sum_{\alpha \in \mathbb{N}^d} \hat{f}(\alpha) \hat{\xi}(\alpha) = D \cdot \Omega_{\mathrm{var}}^*(\xi) .
\end{equation*}
Now since~$\Omega_{\mathrm{var}}$ is the sum of~$d+1$ semi-norms~$\Theta_0, \Theta_1,\dots,\Theta_d$, with 
\begin{align*}
&\Theta_0(f) = |\hat{f}(0)| \\
& \Theta_a(f) = \bigg(\sum_{\alpha \in (\mathbb{N}^d)^*} \frac{\alpha_a}{c_{|\alpha|}} \hat{f}(\alpha)^2\bigg)^{1/2}, \forall a \in [d],
\end{align*}
we have
\begin{equation*}
\Omega_{\mathrm{var}}^*(\xi)
= \inf_{ \xi = \sum_{a=0}^d \xi_a} \sup_{a \in \{0, \ldots, d\}} \Theta_a^*(\xi_a).
\end{equation*}
This is an extension of the fact that the set~$\Omega_{\mathrm{var}}^*(\xi) \leq 1$ is the subdifferential of~$\Omega_{\mathrm{var}}$ at~$f=0$, and thus the sum of the~$d$ subdifferentials of~$\Omega_0,\dots,\Omega_d$ at~$f=0$. We consider $a \in [d]$, we have
\begin{align*}
\Omega_a^*(\xi_a)^2 = \sum_{\alpha \in \mathbb{N}^d, \alpha_a > 0} \hat{\xi}_a(\alpha) ^2 \frac{c_{|\alpha|}}{\alpha_a},
\end{align*}
and $\Omega_0^*(\xi)^2 = \hat{\xi}(0)^2$.
 
If we choose~$\forall \alpha \in (\mathbb{N}^d)^*, \ \hat{\xi}_a(\alpha) = \frac{\sqrt{\alpha_a}}{\sum_b \sqrt{\alpha_b}}\hat{\xi}(\alpha)$, $\hat{\xi}_0(\alpha) = 0$, $\hat{\xi}_a(0) =0$ and $ \hat{\xi}_0(0) = \hat{\xi}(0)$, we have:
\begin{align*}
  \Omega_{\mathrm{var}}^\ast(\xi)^2   & \leq \sup \bigg( \sup_{a \in [d]}   \sum_{\alpha \in \mathbb{N}^d, \alpha_a >0} \hat{\xi}_a(\alpha) ^2 \frac{c_{|\alpha|}}{\alpha_a} , \hat{\xi}_0(0)^2 \bigg) \\
   & \leq \sup \bigg(  \sup_{a \in [d]} \sum_{\alpha \in \mathbb{N}^d, \alpha_a >0} \hat{\xi}(\alpha) ^2 \frac{c_{|\alpha|}}{\big(\sum_b \sqrt{\alpha_b}\big)^2} , \hat{\xi}(0)^2 \bigg)\\
    & \leq \sum_{\alpha \in \mathbb{N}^d} \hat{\xi}(\alpha)^2 \bigg(\frac{c_{|\alpha|}}{|\alpha|}\mathds{1}_{|\alpha|>0}  + \mathds{1}_{|\alpha| =0} \bigg).
\end{align*}
Let~$W^2 = \Diag\big(\frac{c_{|\alpha|}}{|\alpha|}\mathds{1}_{|\alpha|>0} +  \mathds{1}_{|\alpha| =0}\big)$ and~$\Phi$ the design matrix of all~$H_\alpha(x^{(i)})$ (with~$n$ rows and infinitely many columns indexed with~$\alpha \in \mathbb{N}^d$). We have~$\hat\xi = \frac{1}{n} \Phi^\top \varepsilon$, and
\begin{equation*}
 \Omega_{\mathrm{var}}^*(\xi)^2 \leq \hat{\xi}^\top W^2 \hat{\xi} =   \frac{1}{n^2}  \varepsilon^\top  \Phi W^2 \Phi^\top \varepsilon.
\end{equation*}
We  compute the expectation of~$\Omega_{\mathrm{var}}^*(\xi)^2$ for~$\varepsilon \sim \mathcal{N}(0, I_n)$, and get 
\begin{align*}
    \mathbb{E}_\varepsilon\big(\Omega_{\rm var}^*(\xi)^2\big) \leq \mathbb{E}_\varepsilon\bigg( \frac{1}{n^2} \varepsilon^\top  \Phi W^2 \Phi^\top \varepsilon \bigg) &= \frac{1}{n^2} \trace\big(\Phi W^2 \Phi^\top\big) \\ &= \frac{1}{n} + \frac{1}{n^2}  \sum_{\alpha \in (\mathbb{N}^d)^*} \sum_{i=1}^n \frac{c_{|\alpha|}}{|\alpha|} H_\alpha(x^{(i)})^2.
\end{align*}
We now take expectations with regards to the data~$\mathcal{D}$ and get 
\begin{equation*}
      \mathbb{E}_{\mathcal{D},\varepsilon}\big(\Omega_{\mathrm{var}}^*(\xi)^2\big) \leq \frac{1}{n} \sum_{\alpha \in (\mathbb{N}^d)^*} \frac{c_{|\alpha|}}{|\alpha|} \mathbb{E}_{X}\big( H_\alpha(X)^2 \big) + \frac{1}{n}.
\end{equation*}
Using Cauchy-Schwartz,
$\mathbb{E}_{\mathcal{D}, \varepsilon}\big(\Omega_{\mathrm{var}}^*(\xi)\big) \leq \sqrt{\frac{1}{n} \big( 1+ \sum_{\alpha \in (\mathbb{N}^d)^*} \frac{c_{|\alpha|}}{|\alpha|}\mathbb{E}_{X}\big(H_\alpha(X)^2\big)\big)}$. 

From Lemma~\ref{lem:properties_of_omega_feat}, we have
\begin{equation*}
    \Omega_{\mathrm{feat}}(f) \geq  \inf_{R \in O_d} \Omega_{\mathrm{var}}(f(R \cdot)).
\end{equation*}
Then, for an infinite vector~$\xi$ indexed by~$\mathbb{N}^d$, with~$\hat{\xi}(\alpha) = \frac{1}{n} \sum_{i=1}^n \varepsilon_i H_\alpha(x^{(i)})~$  and~$\xi_R$ an infinite vector indexed by~$\mathbb{N}^d$ with~$ \hat{\xi}_R(\alpha) = \frac{1}{n} \sum_{i=1}^n \varepsilon_i H_\alpha(Rx^{(i)})~$, we have~$\Omega_{\mathrm{feat}}^*(\xi) \leq \sup_{R \in O_d} \Omega_{\mathrm{var}}^*(\xi_R)$. 

Therefore
\begin{equation*}
    \sup_{R \in O_d} \Omega_{\mathrm{var}}^*(\xi_R) \leq  \sup_{R \in O_d}  \frac{1}{n^2}  \varepsilon^\top  \Phi_R W^2 \Phi_R^\top \varepsilon,
\end{equation*}
with~$\Phi_R$ the design matrix of all~$H_\alpha(Rx^{(i)})$ (with~$n$ rows and infinitely many columns indexed with~$\alpha \in \mathbb{N}^d$). Therefore using Lemma~\ref{lem:invariance_property},
\begin{align*}
     \varepsilon^\top  \Phi_R W^2 \Phi_R^\top \varepsilon &= \sum_{i,j} \varepsilon_i \varepsilon_j  \bigg(1 + \sum_{\alpha \in (\mathbb{N}^d)^*} \frac{c_{|\alpha|}}{|\alpha|} H_\alpha(Rx^{(i)})H_\alpha(Rx^{(j)}) \bigg)\\
     & = \sum_{i,j} \varepsilon_i \varepsilon_j \bigg( 1 + \sum_{k=1}^{+\infty} \frac{c_k}{k}\sum_{\alpha \in \mathbb{N}^d, \ |\alpha| = k }H_\alpha(Rx^{(i)})H_\alpha(Rx^{(j)}) \bigg) \\
     & = \sum_{i,j} \varepsilon_i \varepsilon_j \bigg( 1+ \sum_{k=1}^{+\infty} \frac{c_k}{k}\sum_{|\alpha| = k }H_\alpha(x^{(i)})H_\alpha(x^{(j)}) \bigg) = \varepsilon^\top  \Phi W^2 \Phi^\top \varepsilon ,
\end{align*}
which is independent of~$R$, therefore yielding exactly the same result as for~$\Omega_{\mathrm{var}}$ once expectation with regards to~$\varepsilon$ and the data is taken.
\end{proof}

\subsection{Statistical convergence}
To gain insight into the proof technique, we initially establish an expectation-based result for the constrained estimator instead of the regularised estimator. We bound the expected risk of the function that minimises the empirical risk over the set of functions with a bounded norm, in expectation over the dataset. To accomplish this, we use Lemma~\ref{lemma:expectation_sups} and Lemma~\ref{lemma:rademacher_Gaussian}.
\begin{lemma}[\note{Expected risk of constrained estimator}]
\label{prop:constrained_expectation}
Let~\note{$D>\Omega(f^*)$}, let~$f^D = \argmin_{f \in \mathcal{F}, \ \Omega(f) \leq D} \widehat{\mathcal{R}}(f)$, under Assumptions~\ref{ass:sampling}, \ref{ass:problem}.1 and \ref{ass:problem}.2, 
\begin{equation*}
    \mathbb{E}_{\mathcal{D}}\big(\mathcal{R}(f^D)\big) \leq \mathcal{R}(f^*) + \frac{4GD}{\sqrt{n}}\sqrt{\frac{\pi}{2}}\sqrt{ 1 + \sum_{\alpha \in (\mathbb{N}^d)^*} \frac{c_{|\alpha|}}{|\alpha|}\mathbb{E}_{X}\big(H_\alpha(X)^2\big)}.
\end{equation*}
\end{lemma}

\begin{proof}[\note{Proof of Lemma~\ref{prop:constrained_expectation}}]
With~$\mathcal{G}:=\{f \in \mathcal{F}, \  \Omega(f) \leq D \}$, we have the classical decomposition of the excess risk
\begin{align*}
    \mathcal{R}(f^D) - \mathcal{R}(f^*) &= \mathcal{R}(f^D) - \widehat{\mathcal{R}}(f^D) + \widehat{\mathcal{R}}(f^D) - \widehat{\mathcal{R}}(f^*) + \widehat{\mathcal{R}}(f^*) - \mathcal{R}(f^*) \\
    &\leq \mathcal{R}(f^D) - \widehat{\mathcal{R}}(f^D) + \widehat{\mathcal{R}}(f^*) - \mathcal{R}(f^*) \\
    & \leq \sup_{f \in \mathcal{G}} \mathcal{R}(f) - \widehat{\mathcal{R}}(f) + \sup_{f \in \mathcal{G}} \widehat{\mathcal{R}}(f) - \mathcal{R}(f).
\end{align*}

We then take the expectation over the data on both sides and use Lemma~\ref{lemma:expectation_sups} and Lemma~\ref{lemma:rademacher_Gaussian}
\begin{align*}
     \mathbb{E}_{\mathcal{D}}\big(\mathcal{R}(f^D)\big) - \mathcal{R}(f^*) 
    & \leq \mathbb{E}_{\mathcal{D}}\big(\sup_{f \in \mathcal{F}} \mathcal{R}(f) - \widehat{\mathcal{R}}(f) + \sup_{f \in \mathcal{F}} \widehat{\mathcal{R}}(f) - \mathcal{R}(f) \big)\\
    & \leq  4\sqrt{\frac{\pi}{2}}G \cdot G_n(\mathcal{G}) \\
    &\leq \frac{4GD}{\sqrt{n}}\sqrt{\frac{\pi}{2}}\sqrt{ 1 + \sum_{\alpha \in (\mathbb{N}^d)^*} \frac{c_{|\alpha|}}{|\alpha|} \mathbb{E}_{X}\big(H_\alpha(X)^2\big)},
\end{align*}
hence the desired result.
\end{proof}

In addition to the expectation-based result, obtaining a result with high probability for the constrained estimator is also of interest. This is achieved in Lemma~\ref{lemma:constrained}, presented in Appendix~\ref{sec:prop_constrained}, by using McDiarmid's inequality \cite{boucheron2013concentration}. To apply this inequality, an additional assumption is required: the boundedness of the loss (Assumption~\ref{ass:problem}.3). However, the most significant and relevant result is the one obtained for the estimator that minimises the regularised empirical risk. This result is more realistic and imposes the additional requirement of convexity of the loss function.

\begin{theorem}[\note{High-probability bound on expected risk of regularised estimator}]
\label{theo:stat_general}
Under Assumption~\ref{ass:sampling} and Assumptions~\ref{ass:problem}.1, \ref{ass:problem}.2, \ref{ass:problem}.3 with~$D = 2\Omega(f^*)$, \ref{ass:problem}.4, then for any~$\delta \in (0,1)$, with the choice of regularising parameter
\begin{equation*}
    \mu =  \frac{8G}{\sqrt{n}}\sqrt{\frac{\pi}{2}}\sqrt{1 + \sum_{\alpha \in (\mathbb{N}^d)^*} \frac{c_{|\alpha|}}{|\alpha|} \mathbb{E}_{X}\big( H_\alpha(X)^2\big)} + \frac{\ell_\infty2\sqrt{2}}{\Omega(f^*)\sqrt{n}}\sqrt{\log \frac{2}{\delta}},
\end{equation*}
with probability larger than~$1-\delta$
\begin{align*}
    \mathcal{R}(f^\mu) & \leq \mathcal{R}(f^*) \\
    & \hspace*{-1cm}+   \Omega(f^*)\bigg(\frac{16G}{\sqrt{n}}\sqrt{\frac{\pi}{2}}\sqrt{1 + \sum_{\alpha \in (\mathbb{N}^d)^*} \frac{c_{|\alpha|}}{|\alpha|} \mathbb{E}_X\big( H_\alpha(X)^2 \big)}\bigg) + \frac{\ell_\infty4\sqrt{2}}{\sqrt{n}}\sqrt{\log \frac{1}{\delta}} \\
    \text{ and } \ 
    \Omega(f^\mu) &\leq 2\Omega(f^*).
\end{align*}    
\end{theorem}

We now discuss the meaning of Theorem~\ref{theo:stat_general}. The theorem states that with high probability, under the appropriate choice of the regularisation parameter, the norm of the estimator $f^\mu$, $\Omega(f^\mu)$, is bounded by twice the norm of the true regression function $f^*$, $\Omega(f^*)$. We remark that the choice of regularisation parameter depends on $\Omega(f^*)$, however, this is not the case in the bounded setting, see the discussion in Section~\ref{sec:dependence}.  Under Assumption~\ref{ass:feature} (feature learning setting) or Assumption~\ref{ass:variable} (variable selection setting), we know that~$\Omega(f^*)$ does not depend explicitly on~$d$ but only on~$s$, the underlying number of variables or dimension of the linear subspace.

The norm $\Omega(f^*)$ also helps us bound the difference between the expected risk of the estimator $\mathcal{R}(f^\mu)$ and the expected risk of the true regression function $\mathcal{R}(f^*)$. This difference, denoted as $\mathcal{R}(f^\mu) - \mathcal{R}(f^*)$, has a dependency on the number of samples $n$, with a convergence rate of $n^{-1/2}$, as expected for a Lipschitz loss and a well-specified model. However, the dependency on the dimension $d$ of the original data is somewhat concealed in $\sqrt{1 + \sum_{\alpha \in (\mathbb{N}^d)^*} \frac{c_{|\alpha|}}{|\alpha|} \mathbb{E}_X\big( H_\alpha(X)^2 \big)}$. We provide a detailed analysis of this dependency for specific choices of the data distribution $X$ and the sequence $(c_k)_{k>0}$ in Section~\ref{sec:dependence}.

\begin{proof}[\note{Proof of Theorem~\ref{theo:stat_general}}]
The proof is adapted from \cite{francis_book}. Define~$f^{\mu *}$ as the minimiser of~$\mathcal{R}_\mu:= \mathcal{R} + \mu\Omega$ over~$\mathcal{F}$. 
 Now, for~$D >0, \note{\tau} >0$ define the following convex set
\begin{equation*}
    \mathcal{C}_{D, \note{\tau}} = \{f \in \mathcal{F}, \Omega(f) \leq D, \mathcal{R}_\mu(f) - \mathcal{R}_\mu(f^{\mu *}) \leq \note{\tau} \}.
\end{equation*}

It has boundary
\begin{equation*}
    \partial \mathcal{C}_{D, \note{\tau}} = \{ f \in \mathcal{F}, \Omega(f) \leq D, \mathcal{R}_\mu(f) - \mathcal{R}_\mu(f^{\mu *}) = \note{\tau} \},
\end{equation*}
i.e., the second constraint is the saturated one, for well chosen~$D$ and~$\note{\tau}$. This is because, if we consider a~$f$ such that~$\Omega(f) = D$, since  the optimality conditions for~$f^{\mu *}$ give that~$\Omega^*(\mathcal{R}^\prime(f^{\mu *})) \leq \mu$, (with~$\mathcal{R}^\prime$ any subgradient of~$\mathcal{R}$ which necessarily exists because~$\mathcal{R}$ is convex since~$\ell$ is convex) we have
\begin{align*}
 \mathcal{R}_\mu(f) - \mathcal{R}_\mu(f^{\mu *}) &=  \mathcal{R}(f) + \mu\Omega(f) - \mathcal{R}(f^{\mu *}) - \mu\Omega(f^{\mu *})  \\ 
 &\geq \langle \mathcal{R}^\prime(f^{\mu *}),(f - f^{\mu *}) \rangle + \mu \Omega(f) - \mu \Omega(f^{\mu *})  \\
 & \text{ by convexity with } \langle \cdot, \cdot \rangle \text{ associated to } \Omega \\ 
 & \geq - \Omega^*\big(\mathcal{R}^\prime(f^{\mu *})\big)\Omega(f - f^{\mu *})  + \mu \Omega(f) - \mu \Omega(f^{\mu *}) \\ 
 & \text{ by Holder's inequality} \\
& \geq   -\mu\Omega(f - f^{\mu *}) + \mu\Omega(f) - \mu\Omega(f^{\mu *}) \text{ by optimality of } f^{\mu *}\\
& \geq  2\mu\Omega(f) - 2\mu\Omega(f^{\mu *}) \text{ by the triangular inequality } \\
 & \geq 2\mu D - 2 \mu\Omega(f^{\mu *}) \text{ since } \Omega(f)=D,   \\
& \geq  2 \mu\Omega(f^{*})  \text{ by  choosing }  D = 2\Omega(f^{*}) \text{, since } \Omega(f^*) \geq \Omega(f^{\mu *}) \\
& \geq  \note{\tau},   \text{ by choosing } \note{\tau} = \mu\Omega(f^{*}),
\end{align*}
hence the desired result on the active constraint of the boundary. We now fix~$\note{\tau} = \mu\Omega(f^{*})$ and~$D = 2\Omega(f^{*})$.

Now if~$f^\mu$ does not belong to~$\mathcal{C}_{D,\note{\tau}}$, since~$f^{\mu *}$ does,  there is an element~$f$ in the segment~$[f^\mu, f^{\mu *}]$ that belongs to~$\partial \mathcal{C}_{D, \note{\tau}}$, i.e,~$\Omega(f) \leq D$ and~$\mathcal{R}_\mu(f) - \mathcal{R}_\mu(f^{\mu *}) = \note{\tau}$. Because the loss is convex, we have that~$\widehat{\mathcal{R}}_\mu(f) \leq \max\{ \widehat{\mathcal{R}}_\mu(f^\mu), \widehat{\mathcal{R}}_\mu(f^{\mu *}) \} = \widehat{\mathcal{R}}_\mu(f^{\mu *})$. Therefore
\begin{align}
\nonumber
    \note{\tau} = \mathcal{R}_\mu(f) - \mathcal{R}_\mu(f^{\mu *}) &\leq \mathcal{R}_\mu(f) - \mathcal{R}_\mu  (f^{\mu *}) +  \widehat{\mathcal{R}}_\mu(f^{\mu *}) - \widehat{\mathcal{R}}_\mu(f) \\ 
    &\leq \mathcal{R}(f) - \widehat{\mathcal{R}}(f)+ \widehat{\mathcal{R}}(f^{\mu *}) -  \mathcal{R}(f^{\mu *}). \label{eq:epsilon} 
\end{align}

From the proof of Lemma~\ref{lemma:constrained}, for all~$\delta \in (0,1)$
\begin{align*}
 \sup_{f \in \mathcal{F}, \ \Omega(f) \leq D}& \mathcal{R}(f) - \widehat{\mathcal{R}}(f) + \sup_{f \in \mathcal{F}, \ \Omega(f) \leq D} \widehat{\mathcal{R}}(f) - \mathcal{R}(f) \\
    &\leq \frac{4GD}{\sqrt{n}}\sqrt{\frac{\pi}{2}}\sqrt{1+ \sum_{\alpha \in (\mathbb{N}^d)*} \frac{c_{|\alpha|}}{|\alpha|} \mathbb{E}_{X}\big(H_\alpha(X)^2\big)} +\frac{\ell_\infty2\sqrt{2}}{\sqrt{n}}\sqrt{\log \frac{1}{\delta}}   
\end{align*}
with probability larger than~$1-\delta$.

We apply this to the RHS of Equation~\eqref{eq:epsilon} (as~$\Omega(f) \leq D$ and~$\Omega(f^{\mu *}) \leq D$), which is smaller than~$ \frac{4GD}{\sqrt{n}}\sqrt{\frac{\pi}{2}}\sqrt{1 + \sum_{\alpha \in (\mathbb{N}^d)^*} \frac{c_{|\alpha|}}{|\alpha|} \mathbb{E}_X\big( H_\alpha(X)^2 \big)}  + \frac{\ell_\infty2\sqrt{2}}{\sqrt{n}}\sqrt{\log \frac{1}{\delta}}$ with probability larger than~$1 - \delta$.

Now if~$\note{\tau}$ is such that
\begin{align*}
    \frac{4GD}{\sqrt{n}}\sqrt{\frac{\pi}{2}}\sqrt{1 + \sum_{\alpha \in (\mathbb{N}^d)^*} \frac{c_{|\alpha|}}{|\alpha|} \mathbb{E}_X\big(H_\alpha(X)^2 \big)} + \frac{\ell_\infty2\sqrt{2}}{\sqrt{n}}\sqrt{\log \frac{1}{\delta}} & \leq \note{\tau}, \text{ i.e., }\\
    \Omega(f^{*})\frac{8G}{\sqrt{n}}\sqrt{\frac{\pi}{2}}\sqrt{1 + \sum_{\alpha \in (\mathbb{N}^d)^*} \frac{c_{|\alpha|}}{|\alpha|} \mathbb{E}_X\big( H_\alpha(X)^2 \big)} + \frac{\ell_\infty2\sqrt{2}}{\sqrt{n}}\sqrt{\log \frac{1}{\delta}} & \leq\mu\Omega(f^{*}) \\
     \frac{8G}{\sqrt{n}}\sqrt{\frac{\pi}{2}}\sqrt{1 + \sum_{\alpha \in (\mathbb{N}^d)^*} \frac{c_{|\alpha|}}{|\alpha|} \mathbb{E}_X\big( H_\alpha(X)^2 \big)} + \frac{\ell_\infty2\sqrt{2}}{\sqrt{n}\Omega(f^{*})}\sqrt{\log \frac{1}{\delta}} & \leq \mu
\end{align*}
then~$f^\mu$ belongs to~$\mathcal{C}_{D, \note{\tau}}$ with probability larger than~$1 - \delta$. If we choose~$\mu =\frac{8GD}{\sqrt{n}}\sqrt{\frac{\pi}{2}}\sqrt{1 + \sum_{\alpha \in (\mathbb{N}^d)^*} \frac{c_{|\alpha|}}{|\alpha|} \mathbb{E}_X\big( H_\alpha(X)^2 \big)} + \frac{\ell_\infty2\sqrt{2}}{\Omega(f^*)\sqrt{n}}\sqrt{\log \frac{1}{\delta}}$, then
\begin{align*}
    \mathcal{R}_\mu & (f^\mu) \leq \mathcal{R}_\mu(f^{\mu *}) + \note{\tau} \\
    &  \leq \mathcal{R}_\mu(f^{\mu *}) + \note{\tau} \\
    &  \leq \mathcal{R}_\mu(f^*) + \note{\tau} \\
    &  \leq \mathcal{R}(f^*) + \mu\Omega(f^*) + \note{\tau} \\
    &   \leq \mathcal{R}(f^*) + 2\mu\Omega(f^*) \\
     &  \leq \mathcal{R}(f^*) \\
     & \ \ +   \Omega(f^*)\bigg(\frac{16G}{\sqrt{n}}\sqrt{\frac{\pi}{2}}\sqrt{1 + \sum_{\alpha \in (\mathbb{N}^d)^*} \frac{c_{|\alpha|}}{|\alpha|} \mathbb{E}_X\big(H_\alpha(X)^2 \big)}\bigg) + \frac{\ell_\infty4\sqrt{2}}{\sqrt{n}}\sqrt{\log \frac{1}{\delta}}
\end{align*}
and~$\Omega(f^\mu) \leq D = 2\Omega(f^*)$ with probability larger than~$1-\delta$.
\end{proof}

\subsection{Dependence on problem parameters}
\label{sec:dependence}

As we have seen, Theorem~\ref{theo:stat_general} depends on some quantities we detail now. \note{First, we provide a definition of subgaussian real variables, as given by \cite{Vershynin_2018}.
\begin{definition}[Subgaussian Variables]
\label{def:subgaussian}
Let $Z$ be a real-valued (not necessarily centred) random variable. $Z$ is subgaussian with variance proxy $\sigma^2$ if and only if 
\begin{equation*}
\forall t > 0, \max\left( \mathbb{P}(Z \geq t), \mathbb{P}(Z \leq -t) \right) \leq e^{-\frac{t^2}{2\sigma^2}}.
\end{equation*}
\end{definition}}

\paragraph{Data distribution.}
To begin, we aim to establish an upper bound for the expectation of the squared Hermite polynomials over the covariates.
\begin{lemma}[\note{Analysis of data-dependent terms in Theorem~\ref{theo:stat_general}}]
\label{lemma:expectation_data}
Let~$\alpha \in \mathbb{N}^d$. 
\begin{enumerate}
    \item If~$X \sim \mathcal{N}(0, I_d)$, then
\begin{equation*}
    \mathbb{E}_X(H_\alpha(X)^2) =1.
\end{equation*}
\item If~$X$ is such that~$\|X\|_2 \leq R$ a.s., then
\begin{equation*}
  \mathbb{E}_X(H_\alpha(X)^2) \leq e^{\frac{R^2}{2}}.
\end{equation*}
\item \note{ If~$X$ is such that $\|X\|_2$ is a subgaussian variable with  variance proxy bounded by~$\sigma^2 < 1/(36e)$, then
\begin{equation*}
    \mathbb{E}_X(H_\alpha(X)^2) \leq e^{36e\sigma^2} \leq e.
\end{equation*}}
\end{enumerate}
\end{lemma}
The proof of this lemma is provided in Appendix~\ref{proof:lemma_expectation_data}. \note{Note that independence between the coordinates is not required, except in the first case, which is an illustration of the definition of the Hermite polynomials.} It is worth noting that except in the Gaussian case, the bounds may not be ideal with respect to their dependency on $d$. However, these bounds rely heavily on the bound for Hermite polynomials in Equation~\eqref{eq:hermite_bound}, which is valid for all points on the real line and for all one-dimensional Hermite polynomials. Thus, it is expected that better bounds in expectation are possible.

\paragraph{Choice of~$(c_k)_{k>0}$.}
The quantities in Theorem~\ref{theo:stat_general} are influenced by the design of the penalty, which is determined by the choice of the sequence $(c_k)_{k>0}$. This dependency is observed in $\Omega(f^*)$, $\ell_\infty$, and $\sqrt{1 + \sum_{\alpha \in (\mathbb{N}^d)^*} \frac{c_{|\alpha|}}{|\alpha|} \mathbb{E}_{X}\big( H_\alpha(X)^2\big)}$. It is worth noting that the bounds provided in Lemma~\ref{lemma:expectation_data} do not rely on the specific value of $\alpha$. Therefore, our focus is now on bounding the summation term $\sum_{\alpha \in (\mathbb{N}^d)^*} \frac{c_{|\alpha|}}{|\alpha|}$.

\begin{lemma}[\note{Analysis of terms depending on $(c_k)_{k>0}$ in Theorem~\ref{theo:stat_general}}]
\label{lem:choice_c}
If~$c_{|\alpha|} = \rho^{|\alpha|}$, with~$\rho \in (0,1)$
\begin{equation*}
    \sum_{\alpha \in (\mathbb{N}^d)^*} \frac{c_{|\alpha|}}{|\alpha|} \leq \frac{1}{(1-\rho)^d}
\end{equation*}
and if~$c_{|\alpha|} = \mathds{1}_{|\alpha|\leq M}$,
\begin{equation*}
    \sum_{\alpha \in (\mathbb{N}^d)^*} \frac{c_{|\alpha|}}{|\alpha|} \leq \frac{M+1}{d} \binom{M+d}{M+1}.
\end{equation*}
\end{lemma}

The proof of this result can be found in Appendix~\ref{proof:lemma:choice_c}. By combining the different results, in the case of bounded data, for example, we can derive a corollary of Theorem~\ref{theo:stat_general} as follows
\begin{corollary}[\note{High-probability bound on expected risk of regularised estimator for bounded data}]
\label{cor:bounded}
    Under Assumption~\ref{ass:sampling} and Assumptions~\ref{ass:problem}.1, \ref{ass:problem}.2, \ref{ass:problem}.3 with~$D = 2\Omega(f^*)$, \ref{ass:problem}.4, if $\|X\|_2 \leq R$ a.s., $(c_k)_{k>0} = (\rho^k)_{k>0}$, then for any~$\delta \in (0,1)$, with the choice of regularising parameter
\begin{equation*}
    \mu =  \frac{G}{\sqrt{n}}\sqrt{1 + \frac{e^{R^2/2}}{(1-\rho)^d}}\bigg( 8\sqrt{\frac{\pi}{2}}  + 2\sqrt{2}\sqrt{\log \frac{2}{\delta}}\bigg),
\end{equation*}
with probability larger than~$1-\delta$
\begin{align*}
    \mathcal{R}(f^\mu) & \leq \mathcal{R}(f^*)  +   \Omega(f^*)\frac{G}{\sqrt{n}}\sqrt{1 + \frac{e^{R^2/2}}{(1-\rho)^d}}\bigg( 16\sqrt{\frac{\pi}{2}}  + 4\sqrt{2}\sqrt{\log \frac{2}{\delta}}\bigg) \\
    \text{ and } \ 
    \Omega(f^\mu) &\leq 2\Omega(f^*).
\end{align*} 
\end{corollary}

The proof is provided in Appendix~\ref{proof:cor_bounded}. We note that the choice of the regularisation parameter is independent of the unknown norm $\Omega(f^*)$ or the distribution of $X$, as long as $R$ is known. In the derived bound, the value of $G$ can be independent of $d$ for certain loss functions such as the logistic loss. We observe that $\Omega(f^*)$ does not depend on the dimension $d$, but solely on the number of variables or the dimension of the linear subspace $s$.  It is important to note that the method exhibits a strong dependence on the dimension, which does not overcome the curse of dimensionality. However, this is merely the first step towards solving the multi-index model through regularised empirical risk minimisation, leaving room for future work and improvements.

\section{Numerical study}
\label{sec:num}
In this section, we present the numerical results that demonstrate the behaviour and performance of \textbf{RegFeaL}. The implementation of the estimator, as well as the code to run the experiments, can be accessed online at {\small \url{https://github.com/BertilleFollain/RegFeaL}}. The \textbf{RegFeaL} estimator class is designed to be compatible with the Scikit-learn API \cite{scikit-learn}, ensuring seamless integration with existing machine learning workflows.

\subsection{Setup}
We describe the experiment setup, which includes data simulation, training procedure and metrics for evaluation.

\paragraph{Data.}
In each generated dataset, depending on whether we consider feature learning or variable selection, we construct the linear subspace $P$ differently. In the feature learning case, we sample a matrix from the set of $d \times d$ orthogonal matrices $O_d$ and select its first $s$ columns to form $P$. For variable selection, we simply consider the first $s$ variables to be the relevant ones. \note{Note that while our experiments were conducted with independently generated covariates, our method is invariant to rotations (in the feature learning case) and sign changes of the data (in both feature learning and variable selection). As such, it is robust to potential correlation between the covariates.} The i.i.d dataset $(x^{(i)}, y^{(i)})_{i \in [n]}$ is then generated as follows
\begin{align*}
& X \sim \mathcal{U}\{[-\sqrt{3}, \sqrt{3}]\}^d \\
& f^*(x) = \sin(2 (P^\top x)_1) + \sin(2 (P^\top x)_2),  \forall x \in \mathbb{R}^d \text{ (sinus dataset)} \\
& f^*(x) = (P^\top x)_1 + (P^\top x)_2 - (P^\top x)_1^2 - (P^\top x)_2^2 + 2(P^\top x)_1(P^\top x)_2^3 - 4, \\
& \forall x \in \mathbb{R}^d \text{ (polynomial dataset)} \\
& Y = f^*(X) + \sigma \varepsilon, \ \varepsilon \sim \mathcal{N}(0, 1).
\end{align*}
Each component of $X$ has mean $0$ and variance $1$. Notably, in both datasets, the true regression function $f^*$ depends on $s=2$ linear combinations of the original variables. The importance of the noise can be controlled through the parameter $\sigma$. The test set $(x^{(i)}_{\text{test}}, y^{(i)}_{\text{test}})_{i \in [n_{\text{test}}]}$ is generated in a similar manner as the training set.

\paragraph{Training.}
The loss that we consider is the quadratic loss. We train \textbf{RegFeaL} on the training set with fixed values of~$\lambda$ and~$r$, and we cross-validate on~$\mu$ and~$\rho$. The number of iterations~$n_{\rm iter}$ depends on the experiment. Some of the parameters are the same in all experiments, such as $n_{\rm test}=5000, s=2, \lambda = 10^{-8} / d^{(2-r)/r}, r=0.33$.

The values of the grid used for cross-validation can be found in Appendix~\ref{app:numerical}. The training pipeline differs between Experiment 1 and Experiments 2 and 3.

In Experiment 1, for each parameter tuple $(\rho, \mu)$, we estimate the number of relevant dimensions $\hat{s}$ using $\hat{s} := |\{ a \in [d], (\eta^{\lambda, \mu}_{\text{feat}})_a^{r/(2-r)} \geq 1/d \}|$. Recall that $\eta_a^{r/(2-r)}$, represents the importance of feature $a$, and at initialisation, it is set to $1/d$ for all $a \in [d]$. We then select $\hat{P}$ as the set of $\hat{s}$ eigenvectors of $\Lambda_{\text{feat}}^{\lambda, \mu}$ corresponding to the $\hat{s}$ largest eigenvalues. Finally, we train a final regressor using Multivariate Adaptive Regression Splines (MARS) \cite{mars} on the dataset $(\hat{P}^\top x^{(i)}, y^{(i)})_{i \in [n]}$.

In Experiments 2 and 3, we simply use the output $f^{\lambda, \mu}_{\text{feat}}$ of Algorithm~\ref{alg:reg_feal} as the prediction function. In both cases, the $R^2$ score is used as the evaluation metric, which is described in Equation~\eqref{eq:$R^2$}.

\paragraph{Metrics.}
We evaluate the performance of \textbf{RegFeaL} using two metrics: the $R^2$ score \cite{r2} for regression performance and an adapted Grassmannian distance for feature learning performance.

The $R^2$  score is computed as
\begin{equation}
\label{eq:$R^2$}
   1 - \frac{\sum_{i=1}^{n_{\rm test}} (y^{(i)}_{\rm test} - y^{(i)}_{\rm pred})^2}{\sum_{i=1}^{n_{\rm test}}(y^{(i)}_{\rm test} - \Bar{y}_{\rm test})^2},
\end{equation}
where $\Bar{y}_{\rm test}$ is the mean of the test response values. The $R^2$ score can be computed on both the training and test sets. A score of $1$ indicates the best possible performance, while a score of $-\infty$ indicates the worst performance. A constant estimator that predicts the average response value corresponds to a score of $0$.

For the feature learning score, we compute the Grassmannian distance between the true subspace $P$ and the estimated subspace $\hat{P}$, which corresponds to the $s$ largest eigenvalue for the score computation. Note that the knowledge of~$s$ is only necessary to compute this score and not necessary for training. Note also that this is not the same $\hat{P}$ that was used to retrain \textbf{MARS} in Experiment~1, as the dimension of that one is estimated. The score is defined as 
\begin{align*}
&\| P(P^\top P)^{-1}P^\top - \hat{P}(\hat{P}^\top \hat{P})^{-1}\hat{P}^\top \|^{2}/(2s) \text{ if } s \leq d/2 \\
&\| P(P^\top P)^{-1}P^\top - \hat{P}(\hat{P}^\top \hat{P})^{-1}\hat{P}^\top \|^{2}/(2(d-s)) \text{ if } s > d/2,
\end{align*}
where $s$ is the number of relevant dimensions. The best possible score is $1$, indicating a perfect match between the true and estimated subspaces, while a score of $0$ indicates no correspondence between the subspaces.

In the setting of variable selection, this discussion can be adapted as discussed in Section~\ref{sec:opti}. The omitted details of the experiments can be found in Appendix~\ref{app:numerical}.

\subsection{Results}
We now provide the results of the experiments.

\paragraph{Experiment 1.}
In this experiment, we investigate the dependence on the dimension of the variables $d$ and the number of samples $n$. We perform the training procedure described earlier, including the retraining step using \textbf{MARS}~\cite{mars} on the projected data. We evaluate the performance on both the sinus dataset and the polynomial dataset with noise levels $\sigma = 0.5$ and $\sigma = 2.5$ respectively. For the sinus dataset, we consider both the variable selection and feature learning settings. We conduct a total of $n_{\rm iter} = 5$ iterations, and the grid used for cross-validation can be found in Appendix~\ref{app:numerical}.

To provide a comparison, we also include the performance of the state-of-the-art method \textbf{MAVE} \cite{MAVE}, which is based on local averaging and does not use regularisation. In our implementation, we follow the recommended procedure for \textbf{MAVE}, which involves first training the Outer Product of Gradients (OPG) method to determine the effective dimensionality reduction (e.d.r) space. We use cross-validation to select the underlying dimension of the space and then retrain the model using \textbf{MARS} on the projected data. This allows us to compute the $R^2$ score. For the feature learning score, we compute it based on the learned effective dimensionality reduction (e.d.r) space. Specifically, we choose $s=2$ as the dimension of the subspace to compute the score, following the same approach as \textbf{RegFeaL}. 

Additionally, we include the $R^2$ score for \textbf{Kernel Ridge}, which uses kernel ridge regression with the kernel $k(x,x^\prime) = \sum_{\alpha \in (\mathbb{N}^d)^*} c_{|\alpha|} H_\alpha(x) H_\alpha(x^\prime)$ and the hyperparameter $\lambda$, which we cross-validate over. To provide a comprehensive analysis, we also display the noise level, which represents the best achievable score considering the noise level $\sigma$. We repeat the entire experiment five times, each time with different data, and present the average results with error bars of $\pm \sigma_{\rm exp}/\sqrt{5}$, where $\sigma_{\rm exp}$ is the standard deviation of the scores across the repetitions. The results of the experiment can be found in Figures~\ref{fig:Experiment1_featureFalse}, \ref{fig:Experiment1_featureTrue}, and \ref{fig:Experiment1_featureTrue_polynomial}.

\begin{figure}[htb]
     \centering
     \begin{subfigure}[b]{165.6pt}
         \centering
         \includegraphics[width=165.6pt]{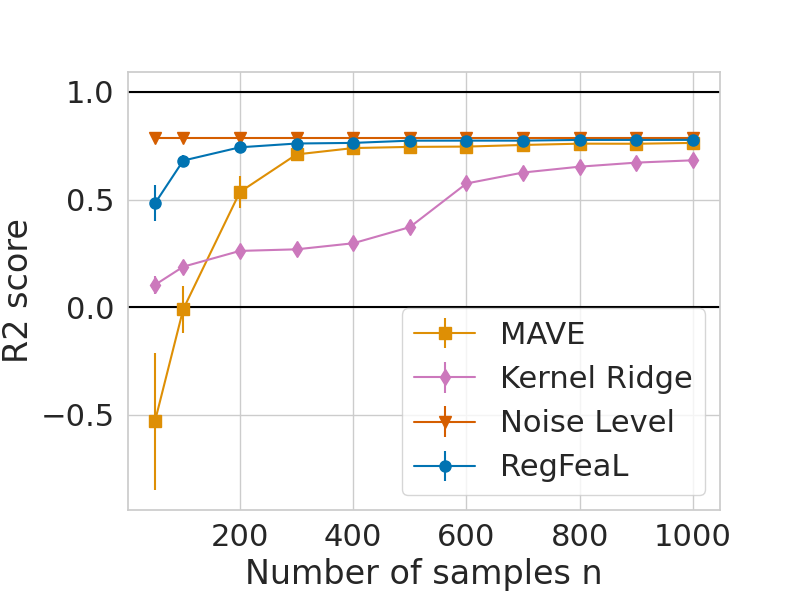}
         \caption{\centering $R^2$ score $d=10$.}
         \label{fig:Experiment1_featureFalse_d10_$R^2$}
     \end{subfigure}
     \hfill
     \begin{subfigure}[b]{165.6pt}
         \centering
         \includegraphics[width=165.6pt]{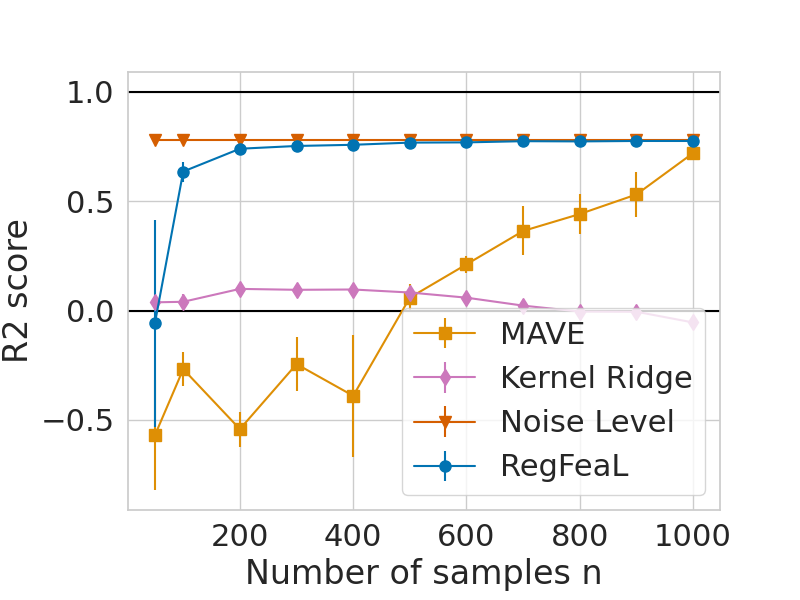}
         \caption{\centering $R^2$ score, $d=40$.}
         \label{fig:Experiment1_featureFalse_d40_$R^2$}
     \end{subfigure}
        \caption{Performance dependency on $d$ and $n$ for the sinus dataset in the variable selection setting.}
        \label{fig:Experiment1_featureFalse}
\end{figure}

\begin{figure}[htb]
     \centering
     \begin{subfigure}[b]{165.6pt}
         \centering
         \includegraphics[width=165.6pt]{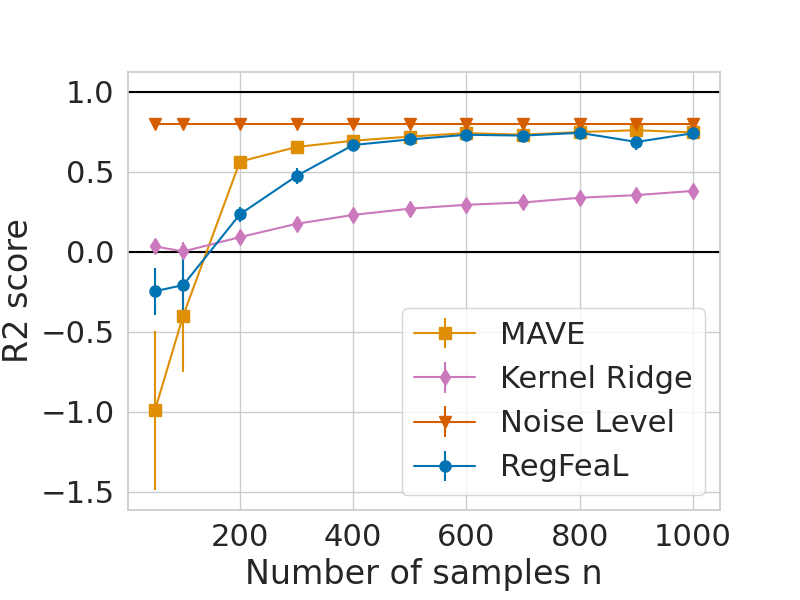}
         \caption{\centering $R^2$ score, $d=10$.}
         \label{fig:Experiment1_featureTrue_d10_$R^2$}
     \end{subfigure}
     \hfill
     \begin{subfigure}[b]{165.6pt}
         \centering
         \includegraphics[width=165.6pt]{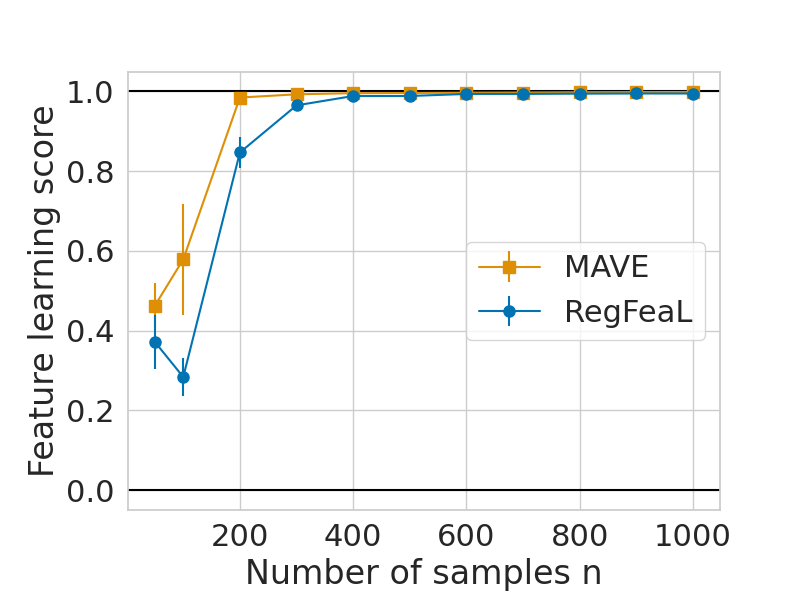}
         \caption{\centering Feature learning score, $d=10$.}
         \label{fig:Experiment1_featureTrue_d10_feature}
     \end{subfigure}
     \begin{subfigure}[b]{165.6pt}
         \centering
         \includegraphics[width=165.6pt]{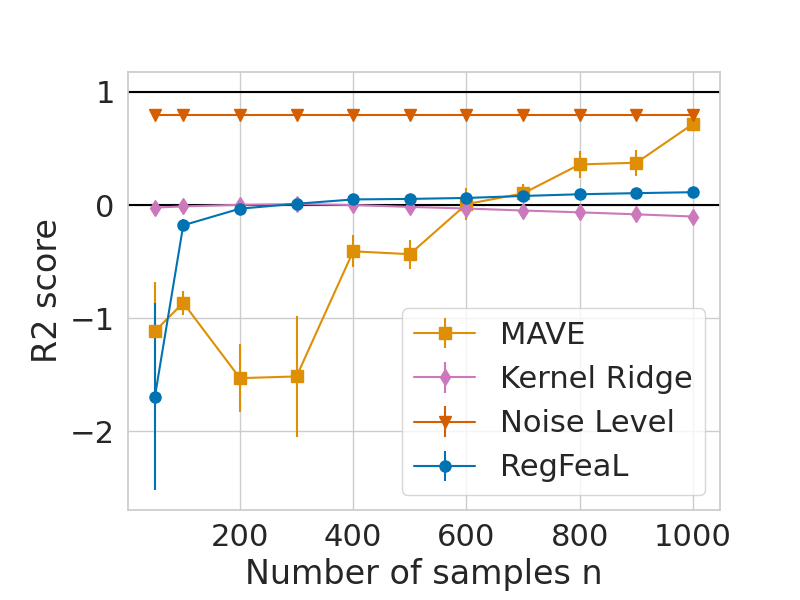}
         \caption{\centering $R^2$ score, $d=40$.}
         \label{fig:Experiment1_featureTrue_d4_$R^2$}
     \end{subfigure}
          \hfill
     \begin{subfigure}[b]{165.6pt}
         \centering
         \includegraphics[width=165.6pt]{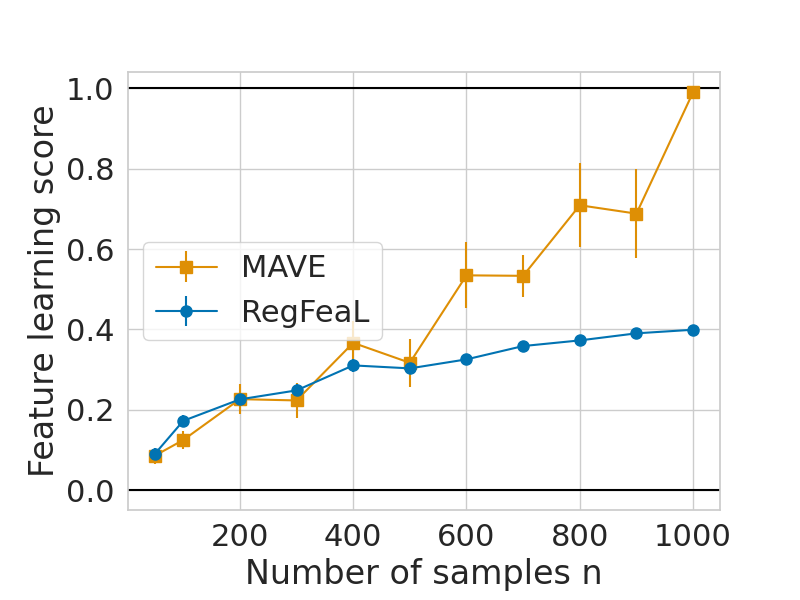}
         \caption{Feature learning score, $d=40$.}
         \label{fig:Experiment1_featureTrue_d40_feature}
     \end{subfigure}
        \caption{Performance dependency on $d$ and $n$ for the sinus dataset in the feature learning setting.}
        \label{fig:Experiment1_featureTrue}
\end{figure}

\begin{figure}[htb]
     \centering
     \begin{subfigure}[b]{165.6pt}
         \centering
         \includegraphics[width=165.6pt]{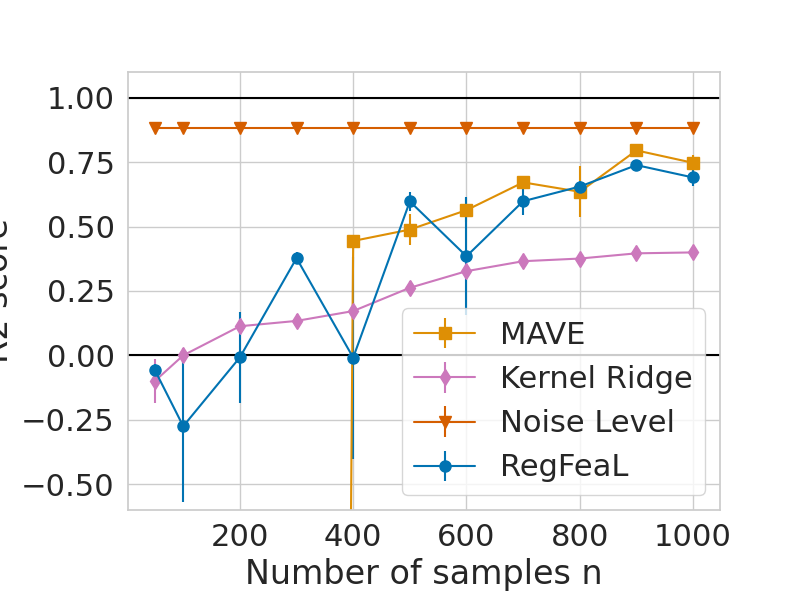}
         \caption{\centering $R^2$ score, $d=20$.}
         \label{fig:Experiment1_featureTrue_d20polynomial_$R^2$}
     \end{subfigure}
     \hfill
     \begin{subfigure}[b]{165.6pt}
         \centering
         \includegraphics[width=165.6pt]{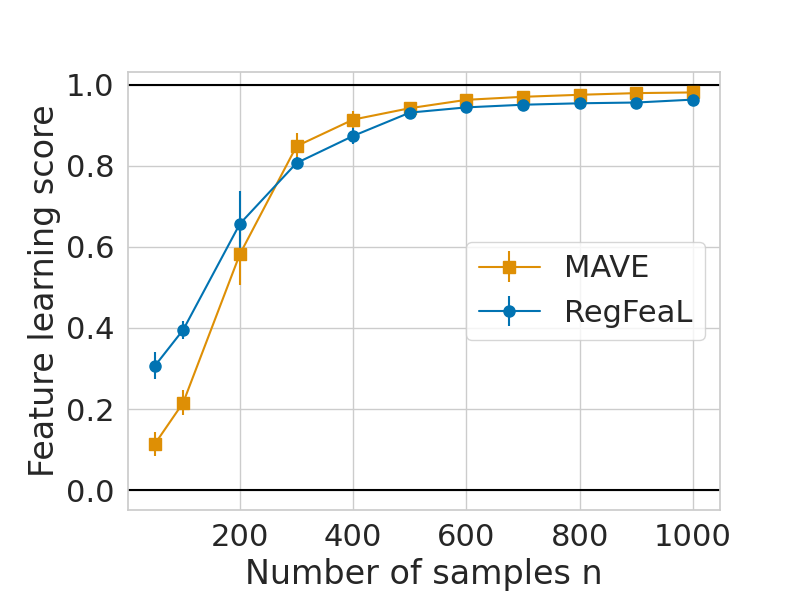}
         \caption{\centering Feature learning score, $d=40$.}
         \label{fig:Experiment1_featureTrue_d20polynomial_feature}
     \end{subfigure}
     \begin{subfigure}[b]{165.6pt}
         \centering
         \includegraphics[width=165.6pt]{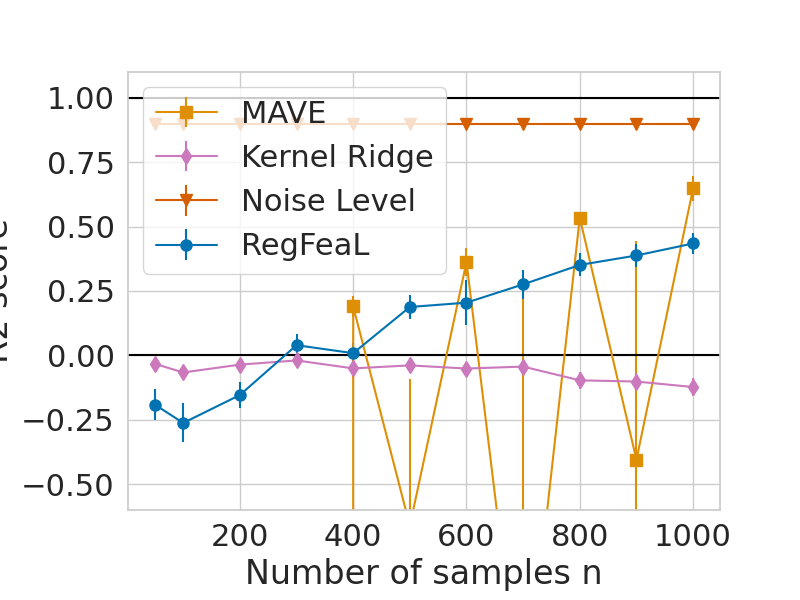}
         \caption{\centering $R^2$ score, $d=10$.}
         \label{fig:Experiment1_featureTrue_d4polynomial_$R^2$}
     \end{subfigure}
          \hfill
     \begin{subfigure}[b]{165.6pt}
         \centering
         \includegraphics[width=165.6pt]{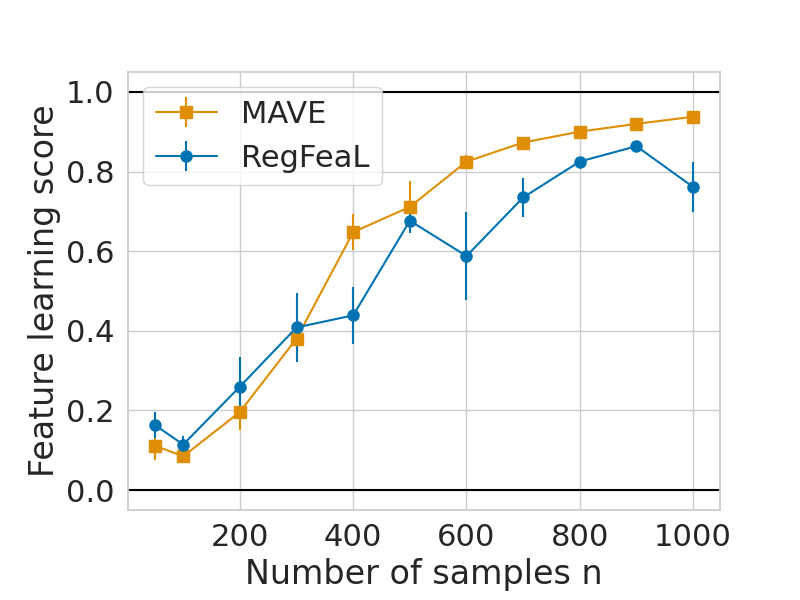}
         \caption{Feature learning score, $d=40$.}
         \label{fig:Experiment1_featureTrue_d40polynomial_feature}
     \end{subfigure}
        \caption{Performance dependency on $d$ and $n$ for the polynomial dataset in the feature learning setting.}
        \label{fig:Experiment1_featureTrue_polynomial}
\end{figure}
In all figures, we observe that the performance improves with a higher number of samples ($n$), which is expected, while it deteriorates with a larger dimension ($d$), which is typical behaviour.

In Figure~\ref{fig:Experiment1_featureFalse}, we focus on the $R^2$ score for the sinus dataset in the variable regression setting. We observe that \textbf{RegFeaL} performs well in both dimensions ($10$ and $40$) without requiring a large number of samples. However, \textbf{Kernel Ridge} fails in dimension $40$ as the kernel cannot effectively capture the dependency on only $2$ variables. As for \textbf{MAVE}, it does not benefit from the knowledge that this is a variable selection problem, unlike \textbf{RegFeaL}, resulting in a higher sample requirement, particularly in dimension $40$.

In Figure~\ref{fig:Experiment1_featureTrue}, we examine the $R^2$ score and the feature learning score for the sinus dataset in the feature learning setting. We observe that \textbf{MAVE} and \textbf{RegFeaL} exhibit similar behaviour in dimension $10$, reaching the noise level for the $R^2$ score and achieving a perfect feature learning score with enough samples. However, in dimension $40$, \textbf{MAVE} struggles significantly when the number of samples is low, while \textbf{RegFeaL} requires a substantially larger sample size to accurately learn the e.d.r. space. Our interpretation is that in this setting, where the true regression function uses a sinus, \textbf{RegFeaL} is hindered by its definition using a basis of polynomials.

In Figure~\ref{fig:Experiment1_featureTrue_polynomial}, we investigate the $R^2$ score and the feature learning score for the polynomial dataset in the feature learning setting. The feature learning performance of \textbf{MAVE} and \textbf{RegFeaL} is similar in this scenario. Regarding the $R^2$ score, \textbf{Kernel Ridge} encounters difficulties in dimension $40$ as it does not benefit from the underlying hidden structure. In dimension $10$, \textbf{RegFeaL} performs similarly to \textbf{MAVE}, but in dimension $40$, it outperforms \textbf{MAVE} as \textbf{MAVE} tends to be overly restrictive and consistently underestimates the number of linear features required to provide a good fit when the e.d.r. space is not perfectly learnt. In contrast, \textbf{RegFeaL} is less conservative, allowing us to leverage more features when the number of samples is too low to accurately estimate them.

\paragraph{Experiment 2.}
In this experiment, we investigate the impact of the number of random features~$m$ (as discussed in Section~\ref{sec:sampling}) on the $R^2$ score and feature learning score for different values of~$n$. The dimension~$d$ is fixed at $10$, while the true underlying dimension~$s$ is $2$. We consider the noiseless setting~$\sigma = 0$ and use the sinus dataset. The same methodology is applied for error bar computation as in Experiment 1. The results are presented in Figure~\ref{fig:experiment2}.

\begin{figure}[htb]
     \centering
     \begin{subfigure}[b]{165.6pt}
         \centering
         \includegraphics[width=165.6pt]{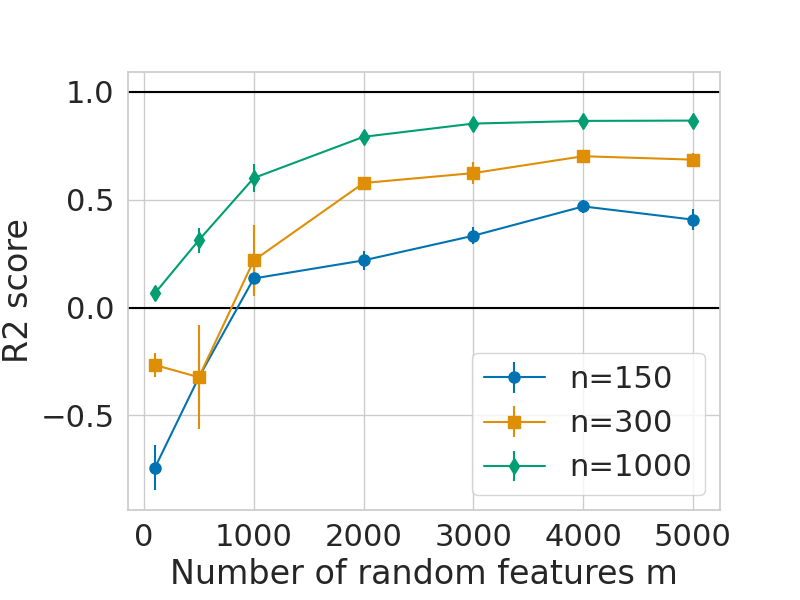}
         \caption{$R^2$ score}
         \label{fig:Experiment2_$R^2$}
     \end{subfigure}
     \hfill
     \begin{subfigure}[b]{165.6pt}
         \centering
         \includegraphics[width=165.6pt]{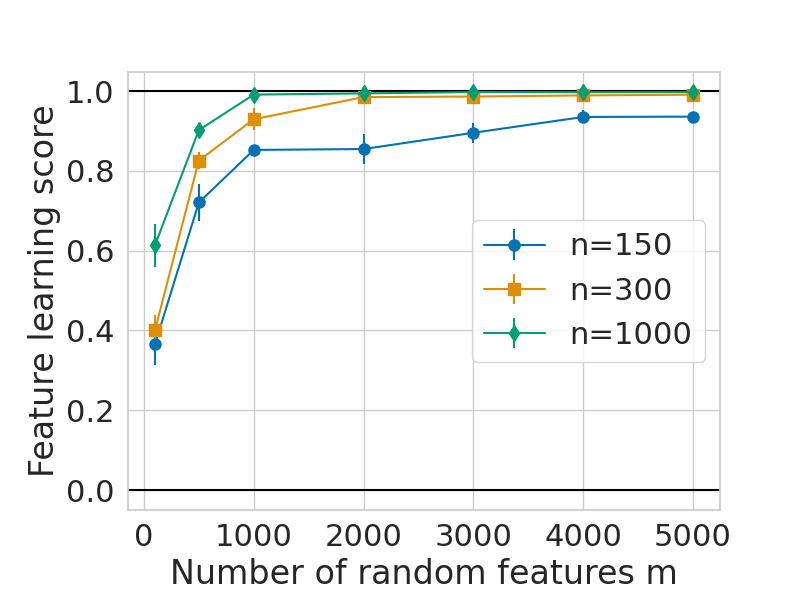}
         \caption{Feature learning score}
         \label{fig:Experiment2_feature}
     \end{subfigure}
        \caption{Influence of the number of random features}
        \label{fig:experiment2}
\end{figure}
We observe that both the $R^2$ score and feature learning score improve with an increase in the number of random features~$m$. This observation aligns with the discussion in Section~\ref{sec:sampling}, where a larger value of~$m$ leads to a better approximation of the kernel~$k_\Lambda$, and allows for a wider range of~$\alpha$ and~$H_\alpha$, resulting in enhanced descriptive power and improved fit and prediction of the subspace. However, we note that beyond a certain value of $m$, the performance improvement levels off while computational costs continue to rise. This suggests that choosing excessively large values of $m$ does not provide any significant benefit.

\paragraph{Experiment 3.}
In this experiment, we maintain the number of samples~$n=5000$, the number of random features~$m=2500$, the dimension~$d=10$, and the underlying dimension~$s=2$ fixed. We work with the noiseless sinus dataset, i.e., $\sigma=0$, and examine the training behaviour of \textbf{RegFeaL} over the iterations. We train the model using cross-validation based on the $R^2$ score and set~$n_{\rm iter} = 10$. The results of this experiment are depicted in Figure~\ref{fig:Experiment3}.

\begin{figure}[htb]
     \centering
     \begin{subfigure}[b]{165.6pt}
         \centering
         \includegraphics[width=165.6pt]{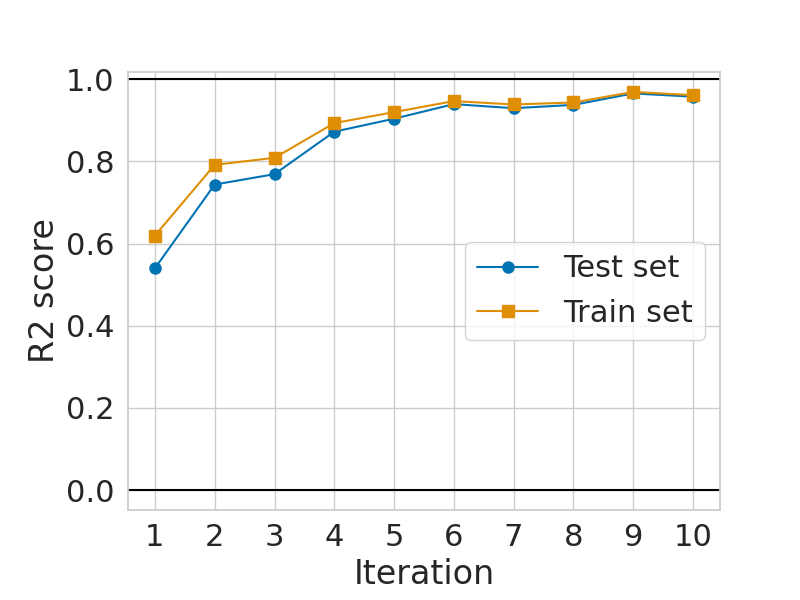}
         \caption{$R^2$ score.}
         \label{fig:Experiment3_$R^2$}
     \end{subfigure}
     \hfill
     \begin{subfigure}[b]{165.6pt}
         \centering
         \includegraphics[width=165.6pt]{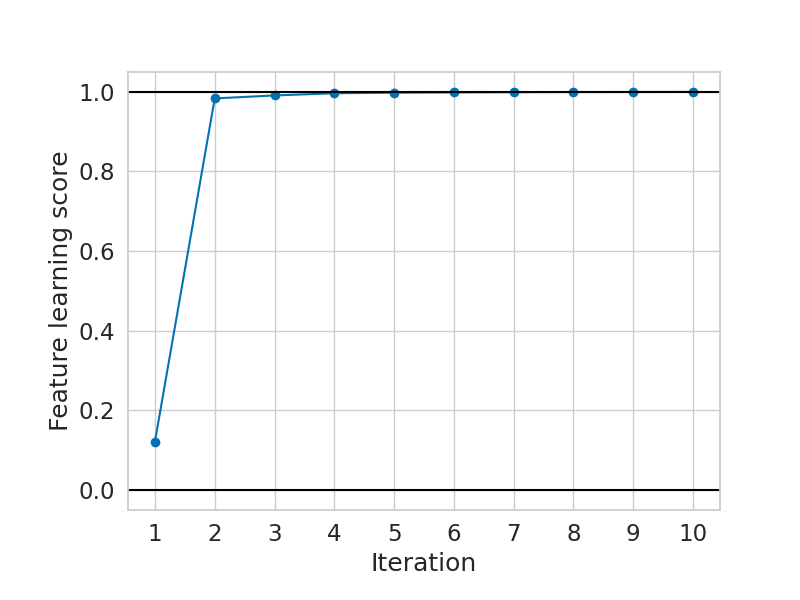}
         \caption{Feature learning score.}
         \label{fig:Experiment3_feature}
     \end{subfigure}
     \begin{subfigure}[b]{165.6pt}
         \centering
         \includegraphics[width=165.6pt]{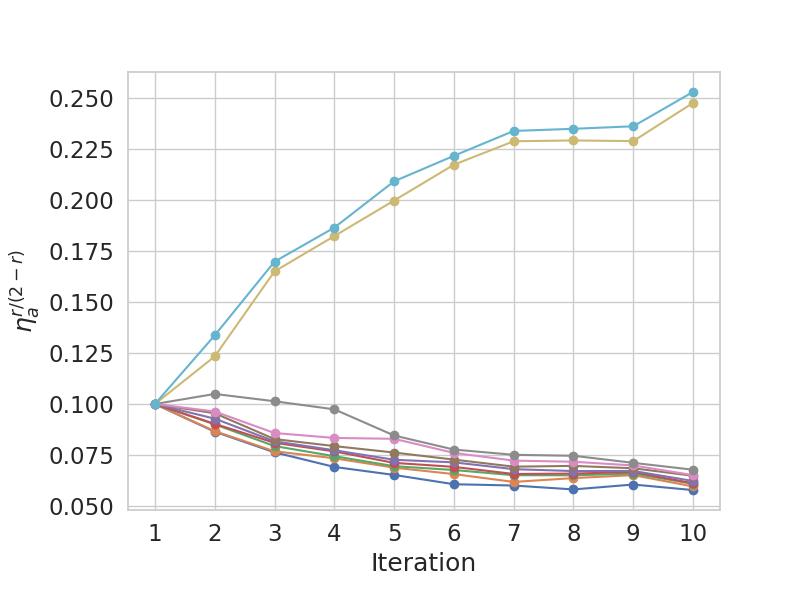}
         \caption{$\eta_a^{r/(2-r)}, \forall a \in [d]$.}
         \label{fig:Experiment3_eta}
     \end{subfigure}
          \hfill
     \begin{subfigure}[b]{165.6pt}
         \centering
         \includegraphics[width=165.6pt]{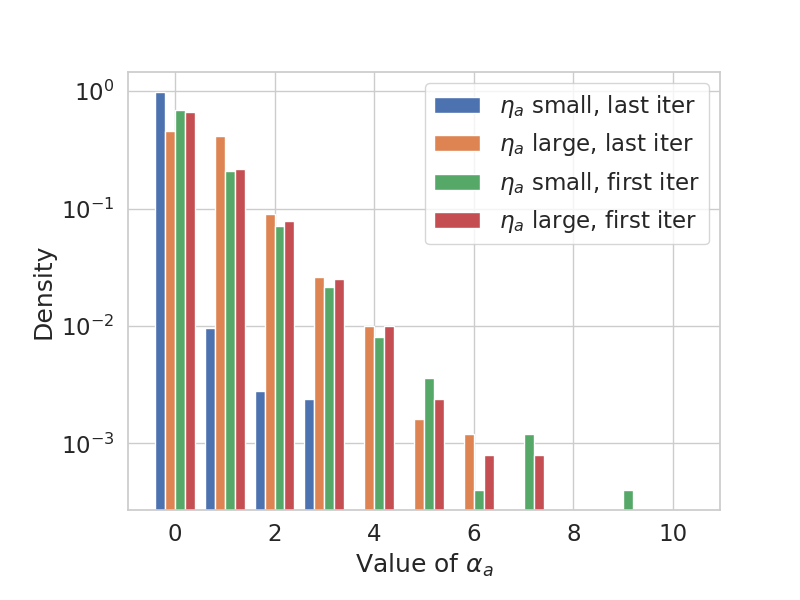}
         \caption{Empirical distribution of~$\alpha$.}
         \label{fig:Experiment3_alpha}
     \end{subfigure}
        \caption{Training behaviour.}
        \label{fig:Experiment3}
\end{figure}
In Figure~\ref{fig:Experiment3_$R^2$}, we observe that the $R^2$ score improves across the iterations on both the test set and the training set. However, the behaviour is not strictly increasing on the training set. This can be attributed to the fact that the kernel approximation differs at each iteration, leading to variations in the fit.

Figure~\ref{fig:Experiment3_feature} demonstrates that the features are learned more rapidly than the fit. It is important to note that the feature learning score assumes knowledge of the underlying dimension~$s=2$. Hence, an important question is whether the estimated value of $s$ is accurate.

In Figure~\ref{fig:Experiment3_eta}, we observe the values of~$\eta_a^{r/(2-r)}$ for all $a \in [d]$ across the iterations. Recall that~$\sum_{a=1}^d \eta_a^{r/(2-r)} = 1$ and that~$\eta_a^{r/(2-r)}$ represents the relative importance of feature~$(R^\top x)_a$. Initially, all~$\eta_a^{r/(2-r)}$ are equal to $1/d$. As the training progresses, most of the components of~$\eta^{r/(2-r)}$ decrease, while two components increase, surpassing the threshold of $1/d$. These two components correspond to the relevant dimensions, indicating that the correct number of dimensions would be easily predicted. Additionally, we observe that these two components of~$\eta$ have relatively similar values, which aligns with the symmetry of the regression function in this example.

Figure~\ref{fig:Experiment3_alpha} displays the empirical density (in log scale) of $\alpha_a$ for two different values of $a \in [d]$ (specifically, $a_{\rm small} := \argmin_{a \in [d]} \eta_a$ and $a_{\rm large} := \argmax_{a \in [d]} \eta_a$ for the final $\eta$) at two different iterations: the first and last iteration. During the first iteration, the distributions of $\alpha_a$ for $a_{\rm large}$ and $a_{\rm small}$ are equal, which aligns with the initialisation discussed in Section~\ref{sec:sampling} (all components of $\eta$ are equal). However, at the end of the optimisation, we observe that the distribution of $\alpha_{a_{\rm small}}$, corresponding to a non-important linear feature, remains almost constant at $0$. Conversely, the distribution of $\alpha_{a_{\rm large}}$, representing an important linear feature, is more widely spread, which is beneficial to the fit.

\section{Conclusion}
\label{sec:conclusion}
We addressed the challenge of \note{prediction function estimation} in multi-index models by proposing a novel approach \textbf{RegFeaL}. Our method combines empirical risk minimisation with derivative-based regularisation to simultaneously estimate the prediction function, the relevant linear transformation, and its dimension. By leveraging the orthogonality and rotation invariance properties of Hermite polynomials, \textbf{RegFeaL} captures the underlying structure of the data. Through alternative minimisation, we iteratively rotate the data to better align it with the leading dimensions.

Theoretical results support the statistical consistency of \note{the expected risk of} our estimator and provide explicit rates of convergence. We demonstrated the performance and effectiveness of our method through extensive empirical experiments on diverse datasets. One of the strengths of our approach is that it does not rely on strong assumptions about the distribution shape or prior knowledge of the subspace dimension.

However, we acknowledge that our method is still subject to the curse of dimensionality, as indicated by the theoretical results showing an exponential dependence on the dimension of the \note{covariates}. Nonetheless, we believe that our findings will contribute to further developments in representation learning and high-dimensional data analysis. Regularisation is a versatile approach that can be applied to a wide range of problems where \note{an} empirical risk can be formulated, \note{foregoing the limitations of some methods solely based on the square loss in supervised learning}.

There are several interesting directions for future research. One possibility is exploring alternative bases other than Hermite polynomials. Additionally, investigating more efficient algorithms and strategies for handling high-dimensional data could be valuable. \note{Furthermore, examining the applicability of our approach to various types of problems and datasets would also be worth pursuing.}

\FloatBarrier
\appendix

\section{Additional proofs and results}
\label{app:theory}

\subsection{Proof of Lemma~\ref{lem:invariance_property}}
\label{proof:invariance_property}
\begin{proof}[\note{Proof of Lemma~\ref{lem:invariance_property}}]
We denote by $\mathcal{N}(0, I_d)$ the $d$-dimensional normal distribution with mean $0 \in \mathbb{R}^d$ and covariance matrix $I_d$. For any~$k \in \mathbb{N}, \ x, x^\prime \in \mathbb{R}^d$, using ~$\forall z \in \mathbb{R}, \ h_k(z)  =  \frac{1}{\sqrt{ k!} } \mathbb{E}_{Y \sim \mathcal{N}(0,1)} (z + i Y)^k$ (which can be shown by recurrence), we have
\begin{align*}
\sum_{|\alpha| = k} & H_\alpha(x) H_\alpha(x^\prime)
 =  \sum_{|\alpha| = k}  \prod_{a=1}^d {h_{\alpha_a}(x_a) h_{\alpha_a}(x^\prime_a)}{  } \\
 = & \mathbb{E}_{Y, Y^\prime \sim \mathcal{N}(0, I_d)} \bigg( \sum_{|\alpha| = k}  \prod_{a=1}^d \frac{1}{\alpha_a !} (x_a + i Y_a)^{\alpha_a}  (x^\prime_a + i Y^\prime_a)^{\alpha_a} \bigg) \\
 = & \frac{1}{k!}  \mathbb{E}_{Y,Y^\prime \sim \mathcal{N}(0,I_d)} \big(  \big( x^\top x^\prime - Y^\top Y^\prime + i ( x^\top Y^\prime + Y^\top x^\prime)  \big)^k\big).
\end{align*}
This shows rotational invariance, that is, for any orthogonal matrix~$R \in O_d$, 
\begin{equation*}
\sum_{|\alpha| = k}  {H_\alpha(x) H_\alpha(x^\prime)} = \sum_{|\alpha| = k}  {H_\alpha(R x) H_\alpha(R x^\prime)}.
\end{equation*}
\end{proof}
\subsection{Proof of Lemma~\ref{lemma:expectation_sups}}
\label{proof:expectation_sups}
\begin{proof}[\note{Proof of Lemma~\ref{lemma:expectation_sups}}]
Define~$\mathcal{H} = \{h : (x,y) \in \mathcal{X}\times \mathcal{Y} \to \ell(y, f(x)), \text{ for }  f \in \mathcal{G} \}$. We have that
\begin{align*}
   &\sup_{f \in \mathcal{G}} \big(\mathcal{R}(f) - \widehat{\mathcal{R}}(f) \big) + \sup_{f \in \mathcal{G}} \big(\widehat{\mathcal{R}}(f) - \mathcal{R}(f) \big)  \\
   &= \sup_{h \in \mathcal{H}} \bigg(\mathbb{E}(h(z)) - \frac{1}{n}\sum_{i=1}^n h(z^{(i)}) \bigg) + \sup_{h \in \mathcal{H}} \bigg( \frac{1}{n}\sum_{i=1}^n h(z^{(i)}) - \mathbb{E}(h(z))\bigg).
\end{align*}

We define the Rademacher complexity of the set~$\mathcal{H}$ by
\begin{equation*}
    R_n(\mathcal{H}) = \mathbb{E}_{\mathcal{D}, \varepsilon \sim \big(\mathcal{U}\{-1,1\}\big)^n }\bigg( \sup_{h \in \mathcal{H}} \frac{1}{n}\sum_{i=1}^n \varepsilon_i h(z^{(i)}) \bigg),
\end{equation*}
where~$\varepsilon \sim \big(\mathcal{U}\{-1,1\}\big)^n$ means that each component of~$\varepsilon$ is independent and follows the uniform distribution over the set~$\{-1, 1\}$.

Using Proposition 4.2 from \cite{francis_book}, we obtain
\begin{align*}
     &\mathbb{E}_{\mathcal{D}}\bigg( \sup_{h \in \mathcal{H}} \mathbb{E}(h(z)) - \frac{1}{n}\sum_{i=1}^n h(z^{(i)}) \bigg) \leq 2 R_n(\mathcal{H})\\ 
     &\mathbb{E}_{\mathcal{D}}\bigg(\sup_{h \in \mathcal{H}}   \frac{1}{n}\sum_{i=1}^n h(z^{(i)}) - \mathbb{E}(h(z))\bigg)  \leq  2 R_n(\mathcal{H}).
\end{align*}

Now from Assumption~\ref{ass:problem}.2 and using Proposition 4.3 from \cite{francis_book}
\begin{equation*}
    R_n(\mathcal{H}) \leq G \cdot R_n(\mathcal{G}),
\end{equation*}
with
\begin{equation*}
    R_n(\mathcal{G}) = \mathbb{E}_{\mathcal{D}, \varepsilon \sim \big(\mathcal{U}\{-1,1\}\big)^n }\bigg( \sup_{f \in \mathcal{G}} \frac{1}{n}\sum_{i=1}^n \varepsilon_i f(x^{(i)}) \bigg).
\end{equation*}
 We have from Exercise 4.9 from \cite{francis_book} that~$R_n(\mathcal{G}) \leq \sqrt{\frac{\pi}{2}} G_n(\mathcal{G})$. Combining all inequalities yields the desired result.
\end{proof}

\subsection{Lemma~\ref{lemma:constrained} and its proof}
\label{sec:prop_constrained}
\begin{lemma}
\label{lemma:constrained}
Under Assumption~\ref{ass:sampling}, Assumptions~\ref{ass:problem}.1,~\ref{ass:problem}.2, \ref{ass:problem}.3,  with~$D \geq \Omega(f^*)$, and~$f^D := \argmin_{f \in \mathcal{F}, \ \Omega(f) \leq D} \widehat{\mathcal{R}}(f)$, for any $\delta \in (0,1)$, with probability larger than~$1-\delta$
\begin{equation*}
       \mathcal{R}(f^D) \leq  \mathcal{R}(f^*) + \frac{4GD}{\sqrt{n}}\sqrt{\frac{\pi}{2}}\sqrt{1 + \sum_{\alpha \in (\mathbb{N}^d)^*} \frac{c_{|\alpha|}}{|\alpha|}\mathbb{E}_{X}\big(H_\alpha(X)^2\big)} + \frac{\ell_\infty2\sqrt{2}}{\sqrt{n}}\sqrt{\log \frac{1}{\delta}}.
\end{equation*}
\end{lemma}

\begin{proof}[\note{Proof of Lemma~\ref{lemma:constrained}}]
Define~$\mathcal{G}:= \{ f \in \mathcal{F}, \ \Omega(f) \leq D \}$. We apply McDiarmid's inequality \cite{boucheron2013concentration} to~$ \sup_{f \in \mathcal{G}} \mathcal{R}(f) - \widehat{\mathcal{R}}(f) + \sup_{f \in \mathcal{G}} \widehat{\mathcal{R}}(f) - \mathcal{R}(f)$, which has bounded variation with constant~$4\ell_\infty/n$, yielding that for all~$\delta \in (0,1)$ 
\begin{align*}
\mathbb{P}_{\mathcal{D}} \bigg( &\sup_{f \in \mathcal{G}} \mathcal{R}(f) - \widehat{\mathcal{R}}(f) +  \sup_{f \in \mathcal{G}} \widehat{\mathcal{R}}(f) - \mathcal{R}(f) \leq \\ 
 & \mathbb{E}\big(\sup_{f \in \mathcal{G}}  \mathcal{R}(f) - \widehat{\mathcal{R}}(f) + \sup_{f \in \mathcal{F}} \widehat{\mathcal{R}}(f) - \mathcal{R}(f)\big) 
    +    \frac{\ell_\infty2\sqrt{2}}{\sqrt{n}}\sqrt{\log \frac{1}{\delta}} \bigg)  \geq 1 - \delta.
\end{align*}
We recall that 
\begin{equation*}
     \mathcal{R}(f^D) - \mathcal{R}(f^*) 
    \leq \sup_{f \in \mathcal{G}} \mathcal{R}(f) - \widehat{\mathcal{R}}(f) + \sup_{f \in \mathcal{G}} \widehat{\mathcal{R}}(f) - \mathcal{R}(f) 
\end{equation*}
and from the proof of Lemma~\ref{prop:constrained_expectation}
\begin{align*}
    \mathbb{E}& \bigg(\sup_{f \in \mathcal{G}}\big( \mathcal{R}(f) - \widehat{\mathcal{R}}(f)\big) + \sup_{f \in \mathcal{G}}\big(  \widehat{\mathcal{R}}(f) - \mathcal{R}(f)\big)\bigg) \\
    &\leq \frac{4GD}{\sqrt{n}}\sqrt{\frac{\pi}{2}}\sqrt{ 1 + \sum_{\alpha \in (\mathbb{N}^d)^*} \frac{c_{|\alpha|}}{|\alpha|}\mathbb{E}_{X}\big(H_\alpha(X)^2\big)},
\end{align*}
yielding the final result.
\end{proof}

\subsection{Proof of Lemma~\ref{lemma:expectation_data}}
\label{proof:lemma_expectation_data}
\begin{proof}[\note{Proof of Lemma~\ref{lemma:expectation_data}}] 
The result for centred normal data with identity covariance matrix is by the construction of the Hermite polynomials \cite{hermite_2009}.

If ~$\|X\|_2$ is bounded by~$R$, using the bound from Equation~\eqref{eq:hermite_bound}, we get that
\begin{align*}
   \mathbb{E}_X(H_\alpha(X)^2) \leq \mathbb{E}(e^{\|X\|^2/2}) \leq \mathbb{E}_X\big(e^{R^2/2}\big) \leq e^{\frac{R^2}{2}}.
\end{align*}

\note{
If ~$X$ is such that  $\|X\|$ is subgaussian with variance proxy $\sigma^2$, we know that $\forall \lambda \leq 1/(6\sqrt{2e}\sigma)$, then $\mathbb{E}_X (e^{\|X\|^2\lambda^2}) \leq e^{72 e\lambda^2 \sigma^2}$ \cite[Proposition 2.5.2]{Vershynin_2018}. Therefore, using the bound from Equation~\eqref{eq:hermite_bound}, we have
\begin{align*}
   \mathbb{E}_X(H_\alpha(X)^2) &\leq \mathbb{E}(e^{\|X\|^2/2}) \leq e^{36e\sigma^2} \leq e
\end{align*}}
This concludes the study of~$\mathbb{E}_X\big( H_\alpha(X)^2 \big)$.
\end{proof}

\subsection{Proof of Lemma~\ref{lem:choice_c}}
\label{proof:lemma:choice_c}
\begin{proof}[\note{Proof of Lemma~\ref{lem:choice_c}}]
Using~$d$-dimensional geometric random variables, we know that
\begin{align*}
    \sum_{\alpha \in \mathbb{N}^d} (1-\rho)^d\rho^{|\alpha|}&= 1, \text{ and therefore }
    \sum_{\alpha \in (\mathbb{N}^d)^{*}}\frac{\rho^{|\alpha|}}{|\alpha|}  \leq \frac{1}{(1-\rho)^d}.
\end{align*}
For the other setting, 
\begin{equation*}
     \sum_{\alpha \in (\mathbb{N}^d)^*, |\alpha| \leq M} \frac{1}{|\alpha|} = \sum_{k=1}^{M} \frac{1}{k} \binom{d-1+k}{d-1} \leq \frac{M+1}{d}\binom{M+d}{M+1},
\end{equation*}
which concludes the proof.
\end{proof}

\subsection{Proof of Corollary~\ref{cor:bounded}}
\label{proof:cor_bounded}
\begin{proof}[\note{Proof of Corollary~\ref{cor:bounded}}] 
First, we note from Lemma~\ref{lemma:expectation_data} that for any $\alpha \in \mathbb{N}^d$, we have $\mathbb{E}_X\big( H_\alpha(X)^2 \big) \leq e^{R^2/2}$. Additionally, from Lemma~\ref{lem:choice_c}, we know that $    \sum_{\alpha \in (\mathbb{N}^d)^*} \frac{c_{|\alpha|}}{|\alpha|} \leq \frac{1}{(1-\rho)^d}$.

Next, we aim to improve the use of McDiarmid's inequality by bounding the deviation of $\sup_{f \in \mathcal{F}, \Omega(f) \leq D} \mathcal{R}(f) - \widehat{\mathcal{R}}(f) + \sup_{f \in \mathcal{F}, \Omega(f) \leq D} \widehat{\mathcal{R}}(f) - \mathcal{R}(f)$ when a single data point $(x^{(i)}, y^{(i)})$ is changed to $(\tilde{x}^{(i)}, \tilde{y}^{(i)})$ changing the dataset from $\mathcal{D}$ to $\tilde{\mathcal{D}}$. In the original proof of Theorem~\ref{theo:stat_general}, we used $4l_\infty/n$ as our bound, but we can provide a tighter bound. We write $\widehat{\mathcal{R}}_\mathcal{D}(f)$ to specify the dependency on the dataset. We also write $\mathcal{G} := \{f \in \mathcal{F}, \ \Omega(f) \leq D \}$. Specifically, we have
\note{
\begin{align*}
& \sup_{f \in \mathcal{G}} \mathcal{R}(f) - \widehat{\mathcal{R}}_\mathcal{D}(f) - \sup_{f \in \mathcal{G}} \mathcal{R}(f) - \widehat{\mathcal{R}}_{\tilde{\mathcal{D}}}(f)  \\
& = \sup_{f \in \mathcal{G}} \mathcal{R}(f) - \widehat{\mathcal{R}}_{\tilde{\mathcal{D}}}(f) + \frac{1}{n}\ell(\tilde{y}^{(i)}, f(\tilde{x}^{(i)})) - \frac{1}{n}\ell(y^{(i)}, f(x^{(i)})) - \sup_{f \in \mathcal{G}} \mathcal{R}(f) - \widehat{\mathcal{R}}_{\tilde{\mathcal{D}}}(f)  \\
& \leq  \sup_{f \in \mathcal{G}} \frac{1}{n}\ell(\tilde{y}^{(i)}, f(\tilde{x}^{(i)})) - \frac{1}{n}\ell(y^{(i)}, f(x^{(i)})), 
\end{align*}
and similarly 
\begin{align*}
 \sup_{f \in \mathcal{G}}&  \widehat{\mathcal{R}}_\mathcal{D}(f) - \mathcal{R}(f)  - \sup_{f \in \mathcal{G}} \widehat{\mathcal{R}}_{\tilde{\mathcal{D}}}(f) - \mathcal{R}(f)  \leq  \sup_{f \in \mathcal{G}}  \frac{1}{n}\ell(y^{(i)}, f(x^{(i)})) - \frac{1}{n}\ell(\tilde{y}^{(i)}, f(\tilde{x}^{(i)}). 
\end{align*}
Combining both and taking the argmax functions $f_1$ and $f_2$, we obtain 
\begin{align*}
 \sup_{f \in \mathcal{G}}& \mathcal{R}(f) - \widehat{\mathcal{R}}_\mathcal{D}(f) - \sup_{f \in \mathcal{G}} \mathcal{R}(f) - \widehat{\mathcal{R}}_{\tilde{\mathcal{D}}}(f)  + \sup_{f \in \mathcal{G}} \widehat{\mathcal{R}}_\mathcal{D}(f) - \mathcal{R}(f)  - \sup_{f \in \mathcal{G}} \widehat{\mathcal{R}}_{\tilde{\mathcal{D}}}(f) - \mathcal{R}(f)   \\
 &\leq \frac{1}{n}\ell(\tilde{y}^{(i)}, f_1(\tilde{x}^{(i)}) - \frac{1}{n}\ell(y^{(i)}, f_1(x^{(i)})) +  \frac{1}{n}\ell(y^{(i)}, f_2(x^{(i)})) - \frac{1}{n}\ell(\tilde{y}^{(i)}, f_2(\tilde{x}^{(i)}) \\
    & \leq \frac{G}{n}(|(f_1-f_2)(x^{(i)}| + |(f_1-f_2)(\tilde{x}^{(i)})|) \\
     &\leq \frac{4}{n}G\sup_{f \in \mathcal{F}, \Omega(f) \leq D, x \in \mathbb{R}^d, \|x\|_2 \leq R} |f(x)| \\
 &\leq \frac{4}{n}GD\sup_{x \in \mathbb{R}^d, \|x\|_2 \leq R}\Omega^*((H_\alpha(x))_{\alpha}) \\
 &\leq \frac{4}{n} GD \sup_{x \in \mathbb{R}^d, \|x\|_2 \leq R}\sqrt{1 + \sum_{\alpha \in (\mathbb{N}^d)^*} \frac{c_{|\alpha|}}{|\alpha|} H_\alpha(x)^2} \\
 & \leq \frac{4}{n} GD \sqrt{1 + \frac{e^{R^2/2}}{(1-\rho)^d}}.
\end{align*}
We can obtain the same exact bound for the opposite quantity of
\begin{align*}
    \sup_{f \in \mathcal{G}}& \mathcal{R}(f) - \widehat{\mathcal{R}}_\mathcal{D}(f) - \sup_{f \in \mathcal{G}} \mathcal{R}(f) - \widehat{\mathcal{R}}_{\tilde{\mathcal{D}}}(f)  + \sup_{f \in \mathcal{G}} \widehat{\mathcal{R}}_\mathcal{D}(f) - \mathcal{R}(f)  - \sup_{f \in \mathcal{G}} \widehat{\mathcal{R}}_{\tilde{\mathcal{D}}}(f) - \mathcal{R}(f)
\end{align*}
by using the same arguments.} We use this bound for $D= 2\Omega(f^*)$. The result follows by employing the proof of Theorem~\ref{theo:stat_general}.
\end{proof}

\section{Technical details of the numerical experiments}
\label{app:numerical}
\paragraph{Experiment 1.}
For \textbf{MAVE} and \textbf{RegFeaL}, the \textbf{MARS} final training used the default parameters provided by the py-earth python package ({\small\url{https://contrib.scikit-learn.org/py-earth/}}), except for the maximum degree, which was taken as the estimated dimension for both methods. \textbf{MAVE} was run using the provided CRAN package in R ({\small \url{https://cran.r-project.org/web/packages/MAVE/index.html}}) and the default parameters.

The number of iterations $n_{\rm iter}$ was set to $5$. For \textbf{RegFeaL}, the cross-validation for~$\rho \times \mu$ was done over the grid defined by~$(0.01, 0.05, 0.1, 0.2, 0.4, 0.6, 0.8)$ for~$\rho$ and $(1000, 100, 10, 1, 0.1, 0.01, 0.001)/d^{((2 - r) / r)}$ for $\mu$. 

The cross-validation for \textbf{Kernel Ridge} was done on parameter $\lambda$, with the set of values $(1000, 100, 10, 1, 0.1, 0.01, 0.001)/d^{((2 - r) / r)}$. The score of the noise level was estimated by $1- \frac{n \sigma^2}{\sum_{i=1}^n (y^{(i)}_{\rm test} - \Bar{y}_{\rm test})^2}$.

\paragraph{Experiment 2.}
For each value of~$n$, we used cross-validation for the largest value of~$m$ and then used the selected~$\rho$ and~$\lambda$ for all other values of~$m$. The cross-validation  was done over the grid defined by~$(0.2, 0.4, 0.6, 0.8, 1.0)$ for~$\rho$ and $(100, 1,0.1,0.01,0.001)/d^{((2 - r) / r)}$ for~$\mu$. The number of iterations~$n_{\rm iter}$ was~$3$. 

\paragraph{Experiment 3.}
The cross-validation for~$\rho \times \mu$ was done over the grid defined by~$(0.2, 0.4, 0.6, 0.8, 1.0)$ for~$\rho$ and $(100, 1,0.1,0.01,0.001)/d^{((2 - r) / r)}$ for~$\mu$.

\begin{acks}[Acknowledgements]
The author thanks Lawrence Stewart, Antonin Brossollet and Oumayma Bounou for fruitful discussions related to this work. The authors are grateful to the CLEPS infrastructure from the Inria of Paris for providing resources and support, particularly  Simon Legrand  ({\small\url{https://paris-cluster-2019.gitlabpages.inria.fr/cleps/cleps-userguide/index.html}}). 
This work is funded in part by the French government under the management of Agence Nationale de la Recherche as part of the “Investissements d’avenir” program, reference ANR-19-P3IA-0001 (PRAIRIE 3IA Institute). We also acknowledge support from the European Research Council (grants SEQUOIA 724063 and DYNASTY 101039676).

\end{acks}

\bibliographystyle{imsart-number} 
\bibliography{biblio}      

\end{document}